\documentclass[journal]{IEEEtran}
\IEEEoverridecommandlockouts
\usepackage{cite}
\usepackage{amsmath,amssymb,amsfonts}
\usepackage{bm}
\usepackage{algorithm}
\usepackage{algpseudocode}
\usepackage{float}
\usepackage{amsthm}
\usepackage{graphics}
\usepackage{graphicx}
\usepackage{textcomp}
\usepackage{xcolor}
\usepackage{dsfont}
\usepackage{comment}
\usepackage{rotating}
\usepackage{makecell}

\usepackage{multirow}

\theoremstyle{plain}

\newtheorem{lemma}{Lemma}

\def\BibTeX{{\rm B\kern-.05em{\sc i\kern-.025em b}\kern-.08em
    T\kern-.1667em\lower.7ex\hbox{E}\kern-.125emX}}
    \DeclareMathOperator*{\argmax}{arg\,max}

\begin{document}
\title{Multi-Static Target Detection and Power Allocation for Integrated Sensing and Communication in Cell-Free Massive MIMO
\thanks{Z. Behdad, K. W. Sung, E.~Bj\"ornson, and C. Cavdar are with the Department of Computer Science, KTH Royal Institute of Technology, Kista, Sweden (\{zinatb, sungkw, emilbjo, cavdar\}@kth.se). \"O. T. Demir is with the  Department of Electrical and Electronics Engineering, TOBB University of Economics and Technology, Ankara, Turkey (ozlemtugfedemir@etu.edu.tr).  Results incorporated in this paper received funding from the ECSEL Joint
Undertaking (JU) under grant agreement No 876124. The JU receives support
from the EU Horizon 2020 research and innovation programme and Vinnova
in Sweden. E.~Bj\"ornson was supported by the SUCCESS project from the Swedish Foundation for Strategic Research.}
\author{\IEEEauthorblockN{Zinat Behdad, \"Ozlem Tu\u{g}fe Demir, Ki Won Sung, Emil Bj\"ornson, and Cicek Cavdar} }}
\maketitle
\vspace{-10mm}
\begin{abstract}
This paper studies an integrated sensing and communication (ISAC) system within a centralized cell-free massive MIMO (multiple-input multiple-output) network for target detection. ISAC transmit access points serve the user equipments in the downlink and optionally steer a beam toward the target in a multi-static sensing framework. A maximum a posteriori ratio test detector is developed for target detection in the presence of clutter, so-called target-free signals. Additionally, sensing spectral efficiency (SE) is introduced as a key metric, capturing the impact of resource utilization in ISAC. A power allocation algorithm is proposed to maximize the sensing signal-to-interference-plus-noise ratio while ensuring minimum communication requirements. Two ISAC configurations are studied: utilizing existing communication beams for sensing and using additional sensing beams. The proposed algorithm's efficiency is investigated in realistic and idealistic scenarios, corresponding to the presence and absence of the target-free channels, respectively. Despite performance degradation in the presence of target-free channels, the proposed algorithm outperforms the interference-unaware benchmark, leveraging clutter statistics. Comparisons with a fully communication-centric algorithm reveal superior performance in both cluttered and clutter-free environments. The incorporation of an extra sensing beam enhances detection performance for lower radar cross-section variances. Moreover, the results demonstrate the effectiveness of the integrated operation of sensing and communication 
compared to an orthogonal resource-sharing approach.
\end{abstract}
\section{Introduction}
Radio-based positioning and sensing have gained considerable attention as one of the key enablers in beyond 5G and 6G wireless networks. It is envisioned that integrated sensing and communication (ISAC), also known as joint communication and sensing (JCAS), will become a promising technology in future wireless systems, providing many sensing applications \cite{liu2022integrated}\cite{behravan2022positioning}. It is expected to support various use cases, including remote healthcare, remote monitoring of weather conditions, asset tracking, gesture recognition, autonomous vehicles, and augmented reality (AR). The main idea of such a system is to share infrastructure, resources, and signals between communications and sensing to improve the spectral efficiency (SE) and sensing performance while lowering the hardware and deployment costs \cite{liu2021survey,zhang2021enabling,wild2021joint,behravan2022positioning,liu2022seventy,liu2020joint}. The communication performance could also be enhanced by exploiting the sensing information to estimate the channel and optimize the network steering \cite{mu2021integrated}.
ISAC benefits from current developments in wireless systems, i.e., denser access point (AP) deployment and wider bandwidths \cite{liu2021survey,zhang2021enabling,behravan2022positioning}.

From the network-level perspective, ISAC is expected to provide a ubiquitous radar perception system \cite{pritzker2022transmit,thoma2021joint, huang2022coordinated}, also known as a \emph{perceptive mobile network} \cite{zhang2020perceptive}.
The focus of the research on ISAC has been mainly on waveform design and signal processing (see  \cite{liu2020joint,buzzi2019using,hua2021transmit,pritzker2022transmit,ye2022beamforming,liu2022joint,he2022joint,li2022framework,li2022integrated,9909950,liu2020radar,mu2021integrated,bazzi2022outage}), where ISAC is mostly considered in single-cell scenarios with mono-static sensing, which refers to a co-located transmitter and receiver for sensing.
Bi-static sensing, involving a non-co-located transmitter and receiver, can be employed to dispense with the full-duplex capability requirement in mono-static sensing. Multi-static sensing, which uses multi non-co-located transmitters and receivers, can also be preferable since it provides a diversity gain due to multiple uncorrelated sensing observations at distributed sensing receivers. Furthermore, multi-static sensing offers an additional performance gain as it achieves increased joint transmit/receive beamforming gain by using multiple transmitters/receivers in the network. 

Synchronization is a challenging issue in practical deployment of multi-static sensing since unsynchronized transmitters and receivers lead to ambiguity and performance degradation. This issue has been also studied for the deployment of a cell-free network for communication purposes \cite{larsson2024massive}. Therefore, cell-free massive MIMO (cell-free massive multiple-input multiple-output) is a promising infrastructure \cite{cell-free-book} to implement multi-static sensing in an existing communication system, as we assume the synchronization issue has been already addressed for communication purposes in cell-free networks. This motivates our work. 
To fully realize the advantages of cell-free massive MIMO, it is preferred to use phase-coherent centralized joint processing. This involves implementing transmit/receive processing by utilizing all the channel estimates obtained from each AP and processing them in a central processing unit. The cloud radio access network (C-RAN) architecture facilitates centralization and synchronization among the distributed APs \cite{demir2023oran}. This architecture has recently gained interest in ISAC systems \cite{huang2022coordinated}.  

There is usually a trade-off between sensing and communication performance due to the shared resources (e.g., time, frequency, power, and space). This trade-off can be balanced by exploiting the existing communication signals for sensing purposes or integrating sensing signals by guaranteeing a specific communication performance. Despite the progress made in this area, there is still a need for new processing and resource allocation schemes that can effectively address both communication and sensing requirements. This becomes even more challenging in cell-free ISAC with joint transmission/reception.

This work studies the detection of a single target in a cell-free ISAC system in the presence of clutter and imperfect channel state information (CSI). The clutter refers to the undesired received signals such as reflections from non-target objects. In this paper, the clutter is characterized by the unknown components of the channels between the ISAC transmit and sensing receive APs. 
It can be treated as interference for sensing and will affect the performance of target detection. However, the communication symbols contribute to sensing by the reflected paths toward the target. Therefore, these reflections are not considered as interference. This is why we design the precoding vectors so that we null the interference only for user equipments (UEs). A power allocation algorithm is proposed to effectively maximize the sensing performance while ensuring that the communication performance remains satisfactory. The preliminary results for a simple channel model with perfect CSI for the communication and sensing channels, without considering clutter, have been published in \cite{behdad2022power}.   
\subsection{Related Work}
Table~\ref{tab:literature} presents an overview of some of the research conducted on ISAC.\footnote{The acronyms in the table are given as OFDM (orthogonal frequency-division multiplexing), BF (beamforming), BER (bit error rate), TX (transmitter), RX (receiver), UL (uplink), DL (downlink), DFRC (dual-functional radar and communications), AF (amplify-and-forward), SINR (signal-to-interference-and-noise ratio), SIR (signal-to-interference ratio), V2I (vehicle-to-infrastructure), EKF (extended Kalman filtering), RMSE (root-mean-square error), ESA (effective sensing area), MI (mutual information).} The table highlights that main focus has been on beamforming design and a few papers have investigated the benefits of power control,\footnote{The beamforming techniques employed in cell-free massive MIMO demonstrate near-optimality.  On the other hand, updating of precoding vectors is required for each coherence block and even for individual target locations. This necessitates algorithm execution for each update. However, a more resource-efficient alternative involves updating power coefficients over prolonged intervals, thereby optimizing computational efficiency.} with notable examples including \cite{buzzi2020transmit} and \cite{huang2022coordinated}. Specifically, \cite{buzzi2020transmit} proposes a power allocation strategy for a mono-static sensing scheme in a single-cell ISAC massive MIMO system.
In contrast, \cite{huang2022coordinated} utilizes distributed radar sensing to estimate the target's location by jointly processing the received signals in a central unit. It also proposes a coordinated power allocation approach where a set of distributed transmit APs individually serve their corresponding UEs in a cellular-type network.

\begin{table*}[t!]
\caption{A literature overview on ISAC.}
\centering
{\begin{tabular}{|c|c|c|c|c|c|}
\hline
\multirow{1.5}{*}{\textbf{Category}} & \multirow{1.5}{*}{\textbf{Ref.}} & \multirow{1.5}{*}{\textbf{System model} }& \multirow{1.5}{*}{\textbf{Contributions}} & \multirow{1.5}{*}{\textbf{Sensing metrics}} & \multirow{1.5}{*}{\textbf{Comm. metrics}}\\[1.5mm]
\hline \hline
\multirow{5}{*}{\textbf{Mono-static}} 
& \multirow{1.5}{*}{\cite{ye2022beamforming}} &  \multirow{1.5}{*}{OFDM-JSAC}  & \multirow{1.5}{*}{Low-complexity JSAC BF algorithms} &\multirow{1.5}{*}{ Beampattern} & \multirow{1.5}{*}{SE and BER}\\[1.5mm]
\cline{2-6} 
& \multirow{1.5}{*}{\cite{liu2022joint}} & \makecell{Single-cell narrow-band\\full-duplex MIMO} & Joint TX-RX BF design & \makecell{Beampattern,\\ velocity, range} & \multirow{1.5}{*}{UL/DL rates}\\[2mm]
\cline{2-6} 
& \multirow{1.5}{*}{\cite{he2022joint}} & \makecell{Cellular DFRC with\\full-duplex relay} & \makecell{Joint transceiver design (AF BF design)} & \makecell{SINR, beampattern,\\detection probability} & \makecell{Nulling radar\\interference} \\[2mm]
\cline{2-6} 
& \makecell{\cite{li2022framework}\\ \cite{li2022integrated}\\ \cite{bazzi2022outage}} & \makecell{Single-cell DFRC-MIMO} & \makecell{BF design under different types of \\ communication interference\cite{li2022framework,li2022integrated},\\
BF design under imperfect CSI\cite{bazzi2022outage}} & \makecell{MI\cite{li2022framework,li2022integrated}\\
Beampattern and\\detection probability\cite{bazzi2022outage}} & Rate\\[2mm]
\cline{2-6} 
& \multirow{1.5}{*}{\cite{guo2022performance}} & \makecell{Multi-antenna full-duplex} & \makecell{An alternative successive interference \\cancellation scheme} & \makecell{Estimation rate of distance,\\ direction, and velocity} & Rate\\[2mm]
\cline{2-6} 
& \multirow{1.5}{*}{\cite{buzzi2020transmit} }& \makecell{Single-cell ISAC \\  massive MIMO system} & \makecell{Power allocation to maximize the \\communication fairness and guarantee \\certain radar SIR} & Detection probability & Rate\\[2mm]
\cline{2-6} 
& \makecell{\cite{liu2020radar},\\ \cite{mu2021integrated}} & \makecell{DFRC V2I} & \makecell{Novel DFRC-based framework for vehicle \\ tracking \& power allocation \cite{liu2020radar},\\ predictive BF design\cite{liu2020radar,mu2021integrated}} & \makecell{RMSE for angle and \\ distance tracking\cite{liu2020radar}, \\
angle\cite{mu2021integrated}} & Sum rate\\[2mm]
\cline{2-6} 
& \makecell{\cite{ouyang2022performance},\\ \cite{ouyang2023integrated}} & \makecell{Single-cell ISAC MIMO} & \makecell{MI-based performance analysis} & Sensing rate & \makecell{Outage\\ probability\cite{ouyang2022performance},\\
rate\cite{ouyang2022performance,ouyang2023integrated}} \\[2mm]
\hline
\multirow{3}{*}{\textbf{Bi-static}} 
& \cite{9909950} & Cellular & \makecell{A new JSAC procedure and protocol} & \makecell{ESA,  maximum detectable\\ velocity, range, and velocity} & BER \\[2mm]
\cline{2-6}
& \makecell{\cite{yang2007mimo},\\ \cite{xie2023sensing}} & MIMO Radar & \makecell{Waveform/precoding design} & MI& -- \\
\hline
\multirow{1.5}{*}{\textbf{\shortstack{Distributed\\bi-static}}} 
& \multirow{1.5}{*}{\cite{sakhnini2022target}} & \makecell{Cell-free massive MIMO\\(fully centralized)} & \makecell{Protocol for UL communications and \\ distributed bi-static sensing} & \makecell{Detection  probability} & \makecell{SE} \\[2mm]
\hline
\multirow{5}{*}{\textbf{Multi-static}} 
& \cite{huang2022coordinated} & \makecell{Multi ISAC-TXs \& multi\\ (centralized) sensing-RXs} & \makecell{Power allocation to minimize the total \\TX power} & Location & SINR \\[2mm]
\cline{2-6} 
& \cite{9920954} & \makecell{V2I (with central controller)} & \makecell{Bandwidth allocation and base station \\bselection using RL} & Estimation rate & Rate \\[2mm]
\cline{2-6}
& \multirow{1.5}{*}{\cite{cheng2022coordinated}} & \multirow{1.5}{*}{Multi-antenna Network ISAC} & \multirow{1.5}{*}{Coordinated transmit BF} & \multirow{1.5}{*}{Detection probability} & \multirow{1.5}{*}{SINR} \\[2mm]
\cline{2-6} 
& \cite{demirhan2023cell} & \makecell{Cell-free ISAC MIMO\\(fully centralized)} & \makecell{ISAC BF design} & \makecell{Sensing SNR} & SINR \\[2mm]
\cline{2-6} 
& \makecell{This\\paper} & \makecell{Cell-free massive MIMO\\(fully centralized)} & \makecell{Power allocation to maximize detection \\ probability while meeting comm. SINR} & \makecell{Detection probability,\\ SE} & SINR, SE \\[1.5mm]
\hline
\end{tabular}}
\label{tab:literature}
\vspace{-5mm}
\end{table*}

The proposed power allocation approach aims at minimizing the total transmit power while ensuring a minimum SINR for each UE and meeting the required Cram\'er-Rao lower bound to estimate the location of the target.

To the best of our knowledge, there is limited research on ISAC in cell-free massive MIMO systems. A recent study in \cite{sakhnini2022target} provides valuable insights into target detection analysis and the number of required antennas to achieve a certain detection probability. In contrast to our work, this study considers uplink communication and a distributed bi-static sensing scheme with a single transmit AP. It proposes a protocol using two sensing modes: i) uplink channel estimation and ii) data payload segment of the communication frame. However, the study does not incorporate any power control mechanism.  A joint sensing and communication BF design is proposed for a cell-free MIMO system in \cite{demirhan2023cell}. The proposed BF design is based on a max-min fairness formulation and is compared with two baseline approaches: i) communication-prioritized sensing BF and ii) sensing-prioritized communication BF. The results in \cite{demirhan2023cell} show that the joint design can achieve the communication SINR close to the SINR of the communication-prioritized approach and almost the same sensing SNR achieved by the sensing-prioritized approach. In this work, the authors assume perfect CSI and do not consider the derivation of the target detector. 

The works in \cite{sakhnini2022target} and \cite{demirhan2023cell} neglect the clutter's influence. However, clutter is unavoidable in practical scenarios and acts as interference for sensing which reduces the SINR and degrades the accuracy of target detection. 
One way to deal with clutter is to improve the signal-to-clutter ratio (SCR). Several techniques have been proposed such as: Doppler processing, beamforming, and clutter suppression using space-time adaptive processing (STAP)\cite{melvinprinciples}. However, clutter cannot be fully suppressed by the above-mentioned techniques. Developing advanced target detector can be employed together with such techniques to improve detection performance. Therefore, further research is necessary to evaluate the performance of ISAC in more practical scenarios.
\subsection{Contributions}
Unlike previous works, this paper considers target detection for ISAC in downlink cell-free massive MIMO systems using a multi-static sensing scheme. 
Two sensing scenarios are analyzed: i) only the downlink communication beams are utilized for sensing;
or ii) additional sensing beams with a certain allocated power are jointly transmitted with the communication beams. In the second scenario, we aim to improve the sensing performance by dedicating some resources to sensing.


 The communication performance is characterized by the minimum attained SINR by the UEs, and detection probability under a certain false alarm probability is used as the sensing performance metric. False alarm probability is the probability of that the system detects a target when the actual target is not present. A power allocation strategy is developed to optimize the sensing performance while satisfying SINR constraints for communication UEs, and per-AP transmit power constraints.

To improve sensing performance, we aim to maximize the sensing SINR under the condition that the target is present.\footnote{The optimal design approach would involve optimizing the transmit power coefficients to maximize the probability of detection. However, such an approach poses a significant challenge due to the strong coupling between the test statistics and the power coefficients.} As discussed in \cite[Chap. 3 and 15]{richards2010principles}, given a fixed false alarm probability, detection probability increases with a higher signal-to-noise ratio (SNR).
The main contributions of this paper are as below:
\begin{itemize}
    \item  A system model is proposed to consider the joint impact of unknown sensing channel coefficients and clutter. The latter is characterized by the target-free channels between transmit APs and receive APs. The channel estimation error is considered for the AP-UE channels, and a new SE expression is obtained by accounting for the interference created by the additional sensing beams. 
    
    \item To maximize the sensing SINR while satisfying specific communication SINR constraints a concave-convex programming (CCP)-based power allocation algorithm is proposed.
    To evaluate the effectiveness of the proposed power allocation algorithm, a communication-centric approach, called \textit{comm.-centric}, is used as the baseline algorithm and representative of the existing cell-free massive MIMO systems which are utilized for serving only communication UEs. The objective of the baseline algorithm is to minimize the total power consumption subject to the same constraints as the proposed algorithm.\footnote{In traditional communication systems, minimizing the total power consumption is one of the common objective functions in the optimization problems.} Conversely, the proposed algorithm maximizes the sensing SINR. In addition, we consider a new benchmark based on an orthogonal sharing approach in which the resources (symbols) in one coherence block are dedicated to communication and sensing in an orthogonal manner. Four algorithms are compared numerically:  i) the communication-centric (\textit{comm.-centric}); ii) the proposed ISAC power allocation algorithm ii.a) without additional sensing beams (\textit{ISAC}), ii.b) with dedicated additional sensing beams (\textit{ISAC+S}); and iii) power allocation algorithms based on orthogonal sharing (\textit{OS}).
  \item We introduce the sensing SE as a new sensing performance metric, which takes into account the impact of the sensing SINR as well as the resource block usage for sensing purposes. We compare the performance of the ISAC approach with an orthogonal sharing approach.
    \item We derive the maximum a posteriori ratio test (MAPRT) detector for target detection. The detector optimally fuses received signals from the receive APs during multiple time slots in a centralized manner to detect a target at a known location within a hotspot area. The detector takes into account the presence of clutter, which is not completely canceled out, and the target's unknown radar cross-section (RCS) values. 
    Moreover, we take into account the potential correlation between different arrival and departure directions to/from the target, which has not been considered before. The new target detector requires more advanced processing compared to the one proposed in our previous work \cite{behdad2022power}, which neglected the existence of clutter. We verify numerically that the proposed advanced processing scheme offers significantly improved detection probability. 
\end{itemize} 
The rest of the paper is organized as follows. Section \ref{section2} describes the system model, including the scenario, channel model, and signal model for communication and sensing. Uplink channel estimation and transmit ISAC precoding vectors are presented in Section \ref{section3}, while the downlink communication and SE are provided in Section \ref{section4}. Sections \ref{section5} and \ref{section6} are devoted to the multi-static sensing algorithm and deriving the MAPRT detector, respectively. The proposed power allocation algorithm is described in Section \ref{section7}. The numerical results are discussed in Section \ref{section8}. Finally, Section \ref{section9} concludes the paper.
\vspace{-2mm}
\subsection{Notations} 
The following mathematical notations are used throughout the paper. Scalars, vectors, and matrices are denoted by regular font, boldface lowercase, and boldface uppercase letters, respectively. The superscript $^T$ shows the transpose operation. Complex conjugate and Hermitian transpose are illustrated by the superscripts $^*$, and $^H$, respectively. The operation of the vectorization of a matrix is represented by $\textrm{vec}(\cdot)$ and the block diagonalization operation is denoted by $\textrm{blkdiag}(\cdot)$. 
The determinant, trace, and real parts of a matrix are represented by $\textrm{det}(\cdot)$, $\textrm{tr}(\cdot)$, and $\Re(\cdot)$, respectively. $\textbf{A}\otimes\textbf{B}$ represents the Kronecker product between matrix $\textbf{A}$ and $\textbf{B}$. The absolute value of a scalar is denoted by $\vert \cdot \vert$ while $\Vert \cdot \Vert$ shows the Euclidean norm of a vector. $\textrm{cov}(x_1,x_2)$ denotes the covariance of the random variables $x_1$ and $x_2$. The parameters are listed in Table~\ref{tab:parameters}.

\begin{table}[]
\centering
	\caption{List of parameters.}  \label{tab:parameters}
    
    \begin{tabular}{c|c}
    \textbf{Parameter} & \textbf{Definition} \\
    \hline
    \hline
        $N_{\rm tx}$ & Number of ISAC transmit APs\\
        \hline
        $N_{\rm rx}$ & Number of sensing receive APs\\
        \hline
        $N_{\rm ue}$ & Number of UEs \\
        \hline
        $M$ & Number of antenna elements per AP\\
        \hline
        $P_{\rm tx}$ & Maximum transmit power per AP \\
        \hline
        $ P_{\rm ul}$ & Uplink transmit power\\
        \hline
        $P_{k}$ & Transmit power at transmit AP $k$\\
        \hline
         $\tau_c,\tau_s$ & Length of coherence block and sensing block, respectively \\
         \hline
         $\tau_p$& Number of pilot symbols in one coherence block\\
         \hline
         $s_i[m]$ & Communication symbol for UE $i$ at time instance $m$\\
         \hline
          $s_0[m]$ & Sensing symbol at time instance $m$\\
         \hline
         $\rho_i, \rho_0$ & Power coefficient for UE $i$ and sensing, respectively\\
         \hline
         $\textbf{x}_k$ & Transmitted signal from transmit AP $k$\\
         \hline
         $\textbf{w}_{i,k}$ & Transmit precoding vector for transmit AP $k$ and UE $i$\\
         \hline
         $\textbf{w}_{0,k}$ & Precoding vector for transmit AP $k$ and sensing signal\\
         \hline
         $\alpha_{r,k}$ & Normalized bi-static RCS of the target from AP $k$ to  AP $r$ \\
         \hline
         $\sigma_{\rm rcs}^2$ & Bi-static RCS variance of the target\\
         \hline
         $\sigma_n^2$ & Variance of the receive noise\\
         \hline
         $\textbf{h}_{i,k}$ & Channel vector from transmit AP $k$ to UE $i$\\
         \hline
          $\textbf{h}_{0,k}$ & Channel vector from transmit AP $k$ to the sensing target\\
         \hline
          $\textbf{H}_{r,k}$ & Target-free channel between transmit AP $k$ and receive AP $r$\\
         \hline
         $y_i, \textbf{y}_r$ & Received signal at UE $i$ and AP $r$, respectively\\
         \hline
        
         $P_{\rm fa},P_{\rm d}$ & False alarm and detection probability, respectively\\
         \hline
         $\gamma_c$ & Communication SINR threshold \\
         \hline
    \end{tabular}
    
\end{table}
\vspace{-4mm}
\section{System Model}
\label{section2}
We consider an ISAC system with downlink communication and multi-static sensing in a C-RAN architecture, as illustrated in Fig.~\ref{fig1}. In this system, each AP either serves as an ISAC transmitter or a sensing receiver. The number of transmit and receive APs are denoted by $N_{\rm tx}$ and $N_{\rm rx}$ at a given instant, respectively. Each AP is equipped with $M$ isotropic antennas deployed as a horizontal uniform linear array (ULA). 

All APs are connected via fronthaul links to the edge cloud and fully synchronized. We consider the centralized implementation of cell-free massive MIMO \cite{cell-free-book} for both communication and sensing-related processing. The $N_{\rm tx}$ ISAC transmit APs jointly serve the $N_{\rm ue}$ UEs by transmitting centrally precoded signals. Optionally, an additional sequence of sensing symbols can be integrated into the transmission that shares the same waveform and time-frequency resources with the communication symbols. When this is the case, the transmit APs jointly steer a centralized beam toward a potential target with a known location within a hotspot area. On the other hand, the $N_{\rm rx}$ APs operate as sensing receivers by simultaneously sensing the location of the target to determine whether there is a target or not. All processing is done at the edge cloud in a centralized manner. 

\begin{figure}[tbp]
\centerline{\includegraphics[trim={5mm 5mm 4mm 2mm},clip,
width=0.3\textwidth]{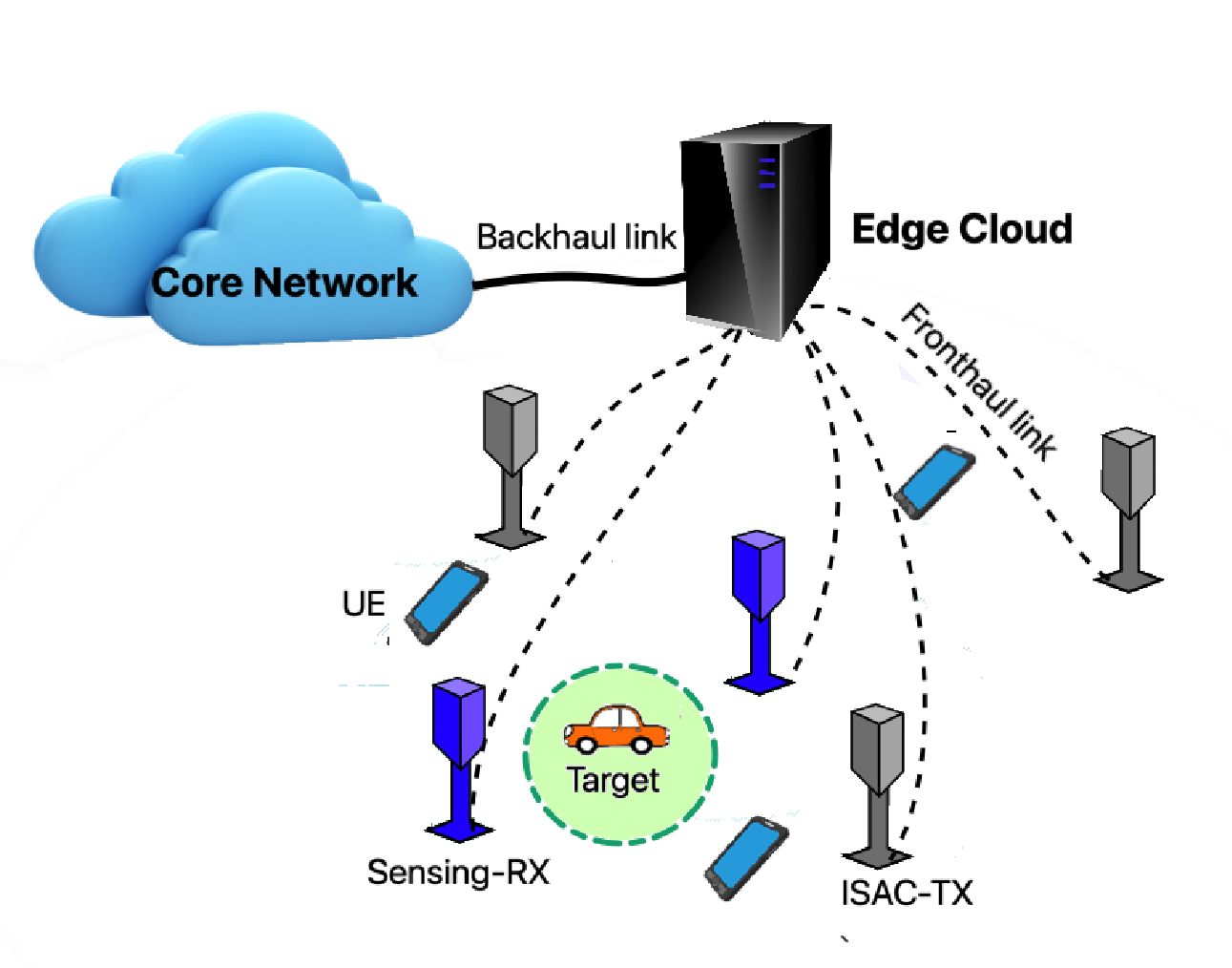}}
\caption{Illustration of the ISAC system setup.}
\label{fig1}
\vspace{-3mm}
\end{figure}
In this system model, $s_i$ represents the zero-mean downlink communication symbol for UE $i$ with unit power, i.e., $\mathbb{E}\{|s_i|^2\}=1$. The sensing signal, denoted by $s_0$, is also assumed to have zero mean and unit power, i.e., $\mathbb{E}\{|s_0|^2\}=1$. We assume that the sensing signal is independent of the UEs' data signals to simplify the analysis and power allocation algorithm\footnote{The sensing signal can instead be selected as a linear combination of UEs' data signals to potentially improve the communication performance for the UEs nearby the hotspot area. We leave the analysis of this case as  future work.}. 
In accordance with \cite{buzzi2019using}, the transmitted signal $\textbf{x}_k[m] \in \mathbb{C}^M$ from transmit AP $k$ at time instance $m$ is
\begin{equation}\label{x_k}
    \textbf{x}_k[m]= \sum_{i=0}^{N_{\rm ue}} \sqrt{\rho_{i}}\textbf{w}_{i,k} s_{i}[m]=\textbf{W}_k \textbf{D}_{\rm s}[m]\boldsymbol{ \rho}, 
\end{equation}
for $k=1,\ldots,N_{\rm tx}$,
where  $\textbf{w}_{i,k}\in \mathbb{C}^{M}$ and $\textbf{w}_{0,k}\in \mathbb{C}^M$ are the transmit precoding vectors for transmit AP $k$ corresponding to UE $i$ and the sensing signal, respectively.  
In \eqref{x_k}, $\textbf{W}_k= \begin{bmatrix}
\textbf{w}_{0,k} & \textbf{w}_{1,k} & \cdots & \textbf{w}_{N_{\rm ue},k}
\end{bmatrix}$, ${\textbf{D}}_{\rm s}[m]=\textrm{diag}\left(s_0[m],s_1[m],\ldots,s_{N_{\rm ue}}[m]\right)$ is the diagonal matrix containing the sensing and communication symbols, and $\boldsymbol{\rho}=[\sqrt{\rho_0} \ \ldots \sqrt{\rho_{N_{\rm ue}}}]^T$.

In line with the C-RAN architecture and the centralized operation of cell-free massive MIMO \cite{cell-free-book}, the precoding vectors for each UE and the target are jointly selected based on the CSI from all the $N_{\rm tx}M$ distributed transmit antennas. Hence, the precoding vectors for each transmit AP are extracted from the concatenated centralized precoding vectors given as
\begin{equation}\label{wi}
        \textbf{w}_i = \begin{bmatrix}
\textbf{w}_{i,1}^T & \textbf{w}_{i,2}^T & \hdots & \textbf{w}_{i,N_{\rm tx}}^T
\end{bmatrix}^T\in \mathbb{C}^{N_{\rm tx}M},
    \end{equation}
    for $i=0,1,\ldots,N_{\rm ue}$.
As seen in \eqref{x_k}, there is a common power control coefficient for each UE, $\rho_i\geq 0$, and for the target, $\rho_0\geq0$, since the precoding is centralized. The rationale is to preserve the direction of the centralized precoding vectors in \eqref{wi} and, hence, not to destroy its favorable characteristics constructed based on the overall channel from $N_{\rm tx}M$ antennas. Using the independence of the data and sensing signals, the average transmit power for transmit AP $k$ is given as
\begin{equation} \label{eq:Pk}
    P_k = 
    \sum_{i=0}^{N_{\rm ue}}\rho_i\mathbb{E}\left\{\Vert \textbf{w}_{i,k} \Vert^2\right\}, \quad k=1,\ldots,N_{\rm tx}.
\end{equation}
Each of the average transmit powers should satisfy the maximum power limit $P_{\rm tx}$, i.e., $P_k\leq P_{\rm tx}$. 

There are three types of channels in this system model: i) communication channels between ISAC transmit APs and the UEs; ii) sensing (target-related) channels between ISAC transmit/sensing receive APs and the target; and iii) interference (target-free) channels between ISAC transmit APs and the sensing receive APs. 
\subsection{Communication Channel Modeling}
We consider block fading communication channels where we have constant and independent channel realizations in each time-frequency coherence block. We let $\textbf{h}_{i,k}\in \mathbb{C}^M$ denote the channel from transmit AP $k$ to UE $i$. We assume a correlated Rayleigh fading distribution for which $\textbf{h}_{i,k}\sim \mathcal{CN}(\textbf{0}, \textbf{R}_{i,k})$ where $\textbf{R}_{i,k}=\mathbb{E}\left\{ \textbf{h}_{i,k} \textbf{h}_{i,k}^H\right\}$ is the spatial correlation matrix that includes the combined effect of geometric path loss, shadowing, and spatial correlation among the antennas for the channel realizations.
We define the communication channel from all the $N_{\rm tx}M$ transmit antennas to UE $i$ in the network as   $\textbf{h}_{i}=\begin{bmatrix}
\textbf{h}_{i,1}^T& \ldots&
\textbf{h}_{i,N_{\rm tx}}^T
\end{bmatrix}^T\in \mathbb{C}^{N_{\rm tx}M}$. 

Imperfect CSI of communication channels is considered, and the standard minimum mean-squared error (MMSE) channel estimator \cite{cell-free-book} is used during the uplink channel estimation phase.

\subsection{Sensing Channel Modeling}
For simplicity, we assume there is a line-of-sight (LOS) connection between the target location and each transmit/receive AP, and the non-line-of-sight (NLOS) connections are neglected. Assuming the antennas at the APs are half-wavelength-spaced, the antenna array response vector for the azimuth angle  $\varphi$ and elevation angle  $\vartheta$, denoted by $\textbf{a}(\varphi,\vartheta)\in \mathbb{C}^M$, is \cite{bjornson2017massive}
\begin{equation}\label{a(phi)}
        \textbf{a}(\varphi,\vartheta) =\begin{bmatrix}
          1& e^{j\pi \sin(\varphi)\cos(\vartheta)}& \ldots& e^{j(M-1)\pi\sin(\varphi)\cos(\vartheta)}
        \end{bmatrix} ^T. 
    \end{equation}

The concatenated array response vector from all the $N_{\rm tx}M$ transmit antennas to the target is given by $\textbf{h}_0 =\begin{bmatrix}
\sqrt{\beta_1}\textbf{a}^T(\varphi_{1},\vartheta_{1})&  \ldots&
\sqrt{\beta_{N_{\rm tx}}}\textbf{a}^T(\varphi_{N_{\rm tx}},\vartheta_{N_{\rm tx}})
\end{bmatrix}^T\in \mathbb{C}^{N_{\rm tx}M}$  using \eqref{a(phi)}. Here $\beta_k$, for $k=1,\dots, N_{\rm tx}$,  is the path loss corresponding to the transmit AP-target path.   
$\varphi_{k}$ and $\vartheta_{k}$ are the azimuth and elevation angles from transmit AP $k$ to the target location, respectively.
\subsection{Target-Free Channel Modeling}
Each receive AP receives both desired and undesired signals when the target is present. The former is the reflected signals from the target, while the latter is independent of the presence of the target.

Apart from the noise, the undesired signal includes clutter and LOS paths between transmit and receive APs. Clutter is composed of two terms: i) the NLOS paths due to scattering/reflection by the permanent objects and ii) the NLOS paths due to the temporary obstacles in the environment. The location and shape of permanent objects such as buildings, walls, or mountains are typically fixed and can be modeled accurately. This means that the LOS and NLOS paths caused by these permanent obstacles can be estimated or measured beforehand using various techniques such as ray tracing, simulation, or empirical models. Therefore, we assume that the LOS and NLOS paths due to the permanent obstacles are known. Assuming the transmit signal $\textbf{x}_k[m]$ is also known at the edge cloud, the known undesired parts of the received signal at each receive AP can be canceled. However, the unknown part corresponding to the NLOS paths caused by temporary obstacles remains.

Let us denote the unknown NLOS channel matrix between transmit AP $k$ and receive AP $r$ by $\textbf{H}_{r,k}\in \mathbb{C}^{M \times M}$. These channels correspond to the reflected paths through the temporary obstacles and are henceforth referred to as target-free channels.
We use the correlated Rayleigh fading model for the NLOS channels $\textbf{H}_{r,k}$, which is modeled using the Kronecker model \cite{Shiu2000a,abramovich2010iterative,zhou2012adaptive,wang2017polarimetric}. Let us define the random matrix $\textbf{W}_{{\rm ch},(r,k)}\in \mathbb{C}^{M \times M}$ with independent and identically distributed (i.i.d.) entries with $\mathcal{CN}(0,1)$ distribution. The matrix $\textbf{R}_{{\rm rx},(r,k)} \in \mathbb{C}^{M \times M}$ represents the spatial correlation matrix corresponding to receive AP $r$ and with respect to the direction of transmit AP $k$. Similarly, $\textbf{R}_{{\rm tx},(r,k)}\in \mathbb{C}^{M \times M}$ is the spatial correlation matrix corresponding to transmit AP $k$ and with respect to the direction of receive AP $r$. Then, the channel $\textbf{H}_{r,k}$ is written as
\begin{equation} \label{eq:Hrk}
    \textbf{H}_{r,k} = \textbf{R}^{\frac{1}{2}}_{{\rm rx},(r,k)} \textbf{W}_{{\rm ch},{(r,k)}}\left(\textbf{R}^{\frac{1}{2}}_{{\rm tx},(r,k)}\right)^T,
\end{equation}
 where the channel gain is determined by the geometric path loss and shadowing, and is included in the spatial correlation matrices. 
\section{Transmit ISAC precoding Vectors}\label{section3}
In this section, we elaborate on the selection of the ISAC precoding vectors based on the estimated UE channels and the location of the sensing point.

The unit-norm regularized zero-forcing (RZF) precoding vector \cite{cell-free-book} is  constructed for UE $i$ as $\textbf{w}_{i}=\frac{\bar{\textbf{w}}_{i}}{\left \Vert \bar{\textbf{w}}_{i}\right \Vert}$, where
\begin{equation}
     \bar{\textbf{w}}_{i}
     \!=\!\left(\sum\limits_{j=1}^{N_{\rm ue}}\hat{\textbf{h}}_j\hat{\textbf{h}}_j^H+\lambda\textbf{I}_{N_{\rm tx}M}\right)^{\!\!-1}\!\!\hat{\textbf{h}}_i, \quad i=1,\ldots,N_{\rm ue},
\end{equation}
and $\lambda$ is the regularization parameter. $\hat{\textbf{h}}_{j}=\begin{bmatrix}
\hat{\textbf{h}}_{j,1}^T& \ldots&
\hat{\textbf{h}}_{j,N_{\rm tx}}^T
\end{bmatrix}^T\in \mathbb{C}^{N_{\rm tx}M}$ 
is the MMSE channel estimate of the communication channel $\textbf{h}_{j}$. Note that since the communication symbols also contribute to sensing by the reflected paths toward the target, they are not considered as interference for the sensing target. Hence, we aim at nulling the interference only for the UEs. To null the destructive interference from the sensing signal to the UEs, the sensing precoding vector $\textbf{w}_{0}$ can be selected as the ZF precoder, i.e., by projecting $\textbf{h}_0$ onto the nullspace of the subspace spanned by the UE channel vectors as in \cite{buzzi2019using}, i.e.,$\textbf{w}_{0}=\frac{\bar{\textbf{w}}_{0}}{\left \Vert \bar{\textbf{w}}_{0}\right \Vert}$, where 
\begin{equation}\label{w0}
   \bar{\textbf{w}_{0}}=\left(\textbf{I}_{N_{\rm tx}M}-\textbf{U}\textbf{U}^{H}\right)\textbf{h}_0,
\end{equation}
where $\textbf{U}$ is the unitary matrix with the orthogonal columns that span the column space of the matrix $\begin{bmatrix}
         \hat{\textbf{h}}_{1}&  \ldots& \hat{\textbf{h}}_{N_{\rm ue}}
        \end{bmatrix}$. 
\section{Downlink Communication and Communication Spectral Efficiency for ISAC}\label{section4}
The received signal at UE $i$ is given as
\begin{align}    \label{y_i}
     y_i[m] =&\sum_{k=1}^{N_{\rm tx}} \textbf{h}^{H}_{i,k}\textbf{x}_{k}[m]+ n_i[m]=\underbrace{\sqrt{\rho_i}\textbf{h}_{i}^{H}\textbf{w}_{i} s_{i}[m]}_{\textrm{Desired signal}}\nonumber \\ &+ \underbrace{\sum_{j=1,j\neq i}^{N_{\rm ue}}\sqrt{\rho_j}\textbf{h}_{i}^{H}\textbf{w}_{j} s_{j}[m]}_{\textrm{Interference signal due to the other UEs}}\nonumber \\ &+ \underbrace{\sqrt{\rho_0}\textbf{h}_{i}^{H}\textbf{w}_{0} s_{0}[m]}_{\textrm{Interference signal due to the sensing}}+ \underbrace{n_i[m]}_{\textrm{Noise}} ,
\end{align}
where the receiver noise at the UE $i$  is represented by $n_i[m] \sim \mathcal{CN}(0,\sigma_n^2)$. We introduce the complex conjugation on the channels $\textbf{h}_{i,k}$ for notational convenience in the downlink, as it is defined in \cite{behdad2022power}. %
For a C-RAN-assisted cell-free massive MIMO system with ISAC, where only $\normalfont\mathbb{E}\left\{\textbf{h}_{i}^{H}\textbf{w}_{i}\right\}$ is known at UE $i$, an achievable SE of this UE is given by
\begin{align} \label{eq:SE_i}
    \mathsf{SE}^{(\rm dl, isac)}_i = \frac{\tau_c-\tau_p}{\tau_c} \log_2\left(1+\mathsf{SINR}^{(\rm dl)}_i\right),
\end{align}
with the downlink SINR as in \eqref{sinr_i}
\begin{figure*}[!t]
\begin{align}\label{sinr_i}
 \mathsf{SINR}^{(\rm dl)}_i=
     & \frac{\rho_i\left\vert \normalfont\mathbb{E}\left\{\textbf{h}_i^H \textbf{w}_i\right\}\right\vert^2}{\sum_{j=1}^{N_{\rm ue}}\rho_j\normalfont\mathbb{E}\left\{\left\vert \textbf{h}_i^H \textbf{w}_j\right\vert^2\right\}-\rho_i\normalfont\left\vert \mathbb{E}\left\{\textbf{h}_i^H \textbf{w}_i\right\}\right\vert^2+\rho_0\normalfont\mathbb{E}\left\{\left\vert \tilde{\textbf{h}}_i^H \textbf{w}_0\right\vert^2\right\}+\sigma_n^2},
     \end{align}
     \hrulefill
     \end{figure*}
 where $\normalfont\tilde{\textbf{h}}_{i}=\begin{bmatrix}
\tilde{\textbf{h}}_{i,1}^T& \ldots&
\tilde{\textbf{h}}_{i,N_{\rm tx}}^T
\end{bmatrix}^T\in \mathbb{C}^{N_{\rm tx}M}$ is the concatenated channel estimation error.
The first two terms in the denominator of \eqref{sinr_i} is related to the interference power due to the other UEs in the system and self interference, while the third one, i.e., $\rho_0\normalfont\mathbb{E}\{\vert \tilde{\textbf{h}}_i^H \textbf{w}_0\vert^2\}$ is the interference power coming from the sensing symbols. Each coherence block contains $\tau_c$ symbols where $\tau_p$ symbols are used as pilots for uplink channel estimation. The pre-log factor represents the fraction of the downlink data within a  coherence block.

The SE given above is achievable when channel coding is applied over a long data sequence during many coherence blocks. In line with this requirement, we assume that we have multiple target locations in the system where sensing each of them takes a portion of the whole communication interval. This implies that the precoding vector $\textbf{w}_0$ changes over time during the whole communication interval according to different locations of the targets and the channels of UEs that vary from coherence block to coherence block. Therefore, the expectation $\mathbb{E}\{\vert \tilde{\textbf{h}}_i^H \textbf{w}_0\vert^2\}$ in \eqref{sinr_i} should be taken over both $\textbf{h}_i$ and  $\textbf{w}_0$ for different channel realizations and sensing locations. This can be done using Monte Carlo simulations.
\section{Multi-Static Sensing}\label{section5}
We consider multi-static sensing where the sensing transmitters and the receivers are not co-located. Apart from the $N_{\rm tx}$ transmit APs, the other $N_{\rm rx}$ receive APs jointly receive the sensing signals for target detection.
We define $\tau_s$ as the number of symbols in one sensing block. In general, we may have either several sensing blocks in each coherence block or a longer sensing block which contains several coherence blocks. Without loss of generality, we consider the former case where $\tau_s \leq \tau_c$. After canceling the known undesired terms, the received signal at AP $r$ in the presence of the target is formulated as
\begin{align}\label{y_rPrim}
          \textbf{y}_r[m] =& \sum_{k=1}^{N_{\rm tx}} \alpha_{r,k}  \sqrt{\beta_{r,k}}\textbf{a}(\phi_{r},\theta_{r})\textbf{a}^{T}(\varphi_{k},\vartheta_{k})\textbf{x}_k[m]\nonumber\\
          &+\sum_{k=1}^{N_{\rm tx}}  \textbf{H}_{r,k}\textbf{x}_k[m]+\textbf{n}_r[m], \quad m=1, \ldots, \tau_s
\end{align}
where ${\textbf{n}_r[m]\sim \mathcal{CN}(\textbf{0},\sigma_n^2\textbf{I}_M)}$ is the receiver noise at the $M$ antennas of receive AP $r$. 
The vector $\textbf{H}_{r,k}\textbf{x}_k[m]$ in \eqref{y_rPrim} can be considered as the interference for the target detection. Moreover, the known part of each reflected path is denoted by ${\textbf{g}_{r,k}[m]\in \mathbb{C}^{M}}$ as
\begin{equation} \label{g_r,k}
   \textbf{g}_{r,k}[m]\triangleq \sqrt{\beta_{r,k}}\textbf{a}(\phi_{r},\theta_{r})\textbf{a}^{T}(\varphi_{k},\vartheta_{k})\textbf{x}_k[m],
\end{equation}
where the matrix  $\alpha_{r,k}\sqrt{\beta_{r,k}} \textbf{a}(\phi_{r},\theta_{r})\textbf{a}^{T}(\varphi_{k},\vartheta_k)$ represents the reflected path through the target. $\phi_r$ and $\theta_r$ are respectively the azimuth and elevation angles from the target location to receive AP $r$. Here,  $\alpha_{r,k}\!\sim\!\mathcal{CN}(0,\,1)$ 
is the normalized bi-static RCS of the target through the reflection path from transmit AP $k$ to receive AP $r$ and is assumed to be unknown. $\beta_{r,k}$ is the channel gain including the path loss from transmit AP $k$ to receive AP $r$ through the target and the variance of bi-static RCS of the target denoted by $\sigma_{r,k}^2$. We follow the Swerling-I model for the RCS, in which the target's velocity is low compared to the total sensing duration. Hence, $\alpha_{r,k}$ is constant throughout the consecutive symbols collected for sensing. We inject the phase-shift of this path to the unknown complex-valued RCS $\alpha_{r,k}$ without loss of generality. 

Let us define $d_{{\rm tx},k}$ as the distance between transmit AP $k$ and target location, $d_{{\rm rx},r}$ as the distance between receive AP $r$ and target location, and $\lambda_c$ as the carrier wavelength. Using the radar range equation for bi-static sensing in \cite[Chap.~2]{richards2010principles}, the channel gain $\beta_{r,k}$ is
\begin{align} \label{eq:radar_range_equation}
    \beta_{r,k} = \frac{\lambda_c^2 \,\sigma_{r,k}^2}{(4\pi)^3 d_{{\rm tx},k}^2 d_{{\rm rx},r}^2}.
\end{align}

Each receive AP sends their respective signals $\textbf{y}_r[m]$, for $r=1,\ldots,N_{\rm rx}$, to the edge cloud. 
    The received signals from the $N_{\rm rx}$ APs can be concatenated to obtain the overall sensing signal
    \begin{equation} \label{y-prime}
  \textbf{y}[m] = \begin{bmatrix}
\textbf{y}_{1}^T[m]&  \cdots&\textbf{y}_{N_{\rm rx}}^T[m]
\end{bmatrix}^T\in \mathbb{C}^{N_{\rm rx}M}.
    \end{equation}
    
To introduce a more concise notation, we concatenate the unknown sensing channel coefficients for each receiver $r$ as $
    \boldsymbol{\alpha}_r \triangleq [ \alpha_{r,1} \ \ldots \ \alpha_{r,N_{\rm tx}} ]^T \in \mathbb{C}^{N_{\rm tx}}$, then collect all  $\boldsymbol{\alpha}_r$ terms in a vector as
$
\boldsymbol{\alpha}= [ \boldsymbol{\alpha}_1^T \ \ldots \ \boldsymbol{\alpha}_{N_{\rm rx}}^T]^T \in \mathbb{C}^{N_{\rm tx}N_{\rm rx}}$. Moreover,  using $\textbf{g}_{r,k}[m]$ from \eqref{g_r,k}, we define the concatenated matrix for receive AP $r$ and symbol $m$ as $\textbf{G}_r[m] = \begin{bmatrix}\textbf{g}_{r,1}[m]& \cdots &\textbf{g}_{r,N_{\rm tx}}[m] \end{bmatrix}\in \mathbb{C}^{M \times N_{\rm tx}}$.

Moreover, we define $\boldsymbol{\mathfrak{h}}_{r,k}= \textrm{vec}\left(\textbf{H}_{r,k}\right)$ as the vectorized unknown target-free channel between receive AP $r$ and the transmit AP $k$ where $\boldsymbol{\mathfrak{h}}_{r,k}\sim \mathcal{CN}\left(\textbf{0},\overline{\textbf{R}}_{r,k}\right) \in \mathbb{C}^{M^2}$. Then, the channel vector $\boldsymbol{\mathfrak{h}}_{r,k}$ can be expressed as
\begin{align}\label{h_r,k}
     \boldsymbol{\mathfrak{h}}_{r,k}&= \left(\textbf{R}^{\frac{1}{2}}_{{\rm tx},(r,k)} \otimes \textbf{R}^{\frac{1}{2}}_{{\rm rx},(r,k)} \right) \textrm{vec}\left(\textbf{W}_{{\rm ch},{(r,k)}}\right) \nonumber\\
     &= \overline{\textbf{R}}_{r,k}^{\frac{1}{2}}\textrm{vec}\left(\textbf{W}_{{\rm ch},{(r,k)}}\right),
 \end{align}
where we have defined $\overline{\textbf{R}}_{r,k}^{\frac{1}{2}} \triangleq \textbf{R}^{\frac{1}{2}}_{{\rm tx},(r,k)} \otimes \textbf{R}^{\frac{1}{2}}_{{\rm rx},(r,k)} $ and used the identity $\mathrm{vec}(\textbf{A}\textbf{B}\textbf{C})=(\textbf{C}^T \otimes \textbf{A})\mathrm{vec}(\textbf{B})$ with $\textbf{A}=\textbf{R}^{\frac{1}{2}}_{{\rm rx},(r,k)}$, $\textbf{B}=\textbf{W}_{{\rm ch},{(r,k)}}$, and $\textbf{C}=\textbf{R}^{\frac{1}{2}}_{{\rm tx},(r,k)}$. We also recognize the square root of the spatial correlation matrix $\overline{\textbf{R}}_{r,k}$ by noting that  $\mathbb{E}\left\{\textrm{vec}\left(\textbf{W}_{{\rm ch},{(r,k)}}\right)\textrm{vec}\left(\textbf{W}_{{\rm ch},{(r,k)}}\right)^H\right\}=\textbf{I}_{M^2}$. Similarly, the interference signal for each receive AP $r$ can be rewritten as 
\begin{align}
    \sum_{k=1}^{N_{\rm tx}}\textbf{H}_{r,k}\textbf{x}_k[m]& =  \sum_{k=1}^{N_{\rm tx}}\left(\textbf{x}^T_k[m] \otimes \textbf{I}_{M} \right)\boldsymbol{\mathfrak{h}}_{r,k}\nonumber\\
    &=\left(\textbf{x}^T[m] \otimes \textbf{I}_{M} \right)\boldsymbol{\mathfrak{h}}_{r},
\end{align}
where we concatenate the unknown and known parts of the interference signals as $\boldsymbol{\mathfrak{h}}_{r}=\begin{bmatrix}\boldsymbol{\mathfrak{h}}_{r,1}^T&\cdots&\boldsymbol{\mathfrak{h}}_{r,N_{\rm tx}}^T\end{bmatrix}^T \in \mathbb{C}^{N_{\rm tx}M^2}$ and $
\textbf{x}[m] = \begin{bmatrix}\textbf{x}^T_1[m]& \cdots &\textbf{x}^{T}_{N_{\rm tx}}[m] \end{bmatrix}^T\in \mathbb{C}^{ N_{\rm tx}M}$, respectively. 

Finally, the overall received signal in \eqref{y-prime} can be expressed as
\begin{align}\label{y_m}
     \textbf{y}[m] &= \underbrace{\mathrm{blkdiag}\left(\textbf{G}_{1}[m],\ldots,\textbf{G}_{N_{\rm rx}}[m]\right)}_{\triangleq \textbf{G}[m]}\boldsymbol{\alpha}\nonumber\\
     &+\underbrace{\left(\textbf{I}_{N_{\rm rx}} \otimes \left(\textbf{x}^T[m] \otimes \textbf{I}_{M}\right)\right)}_{\triangleq \textbf{X}[m]}\boldsymbol{\mathfrak{h}}+\textbf{n}[m],
\end{align}
where $\mathrm{blkdiag}(\cdot)$ constructs a block diagonal matrix, $\boldsymbol{\mathfrak{h}}\!= \!\left[
     \boldsymbol{\mathfrak{h}}_{1}^T \ \!\cdots \ \! \boldsymbol{\mathfrak{h}}^T_{N_{\rm rx}}\right]^T\!\in \!\mathbb{C}^{N_{\rm tx}N_{\rm rx}M^2}$, and $\textbf{n}[m]\!=\!\left[\!\textbf{n}_1^T[m]\! \ \cdots \ \textbf{n}_{N_{\rm rx}}^T[m]\right]^T\!\in\!\mathbb{C}^{N_{\rm rx}M}$ is the concatenated receiver noise.
Given that the channels $\textbf{H}_{r,k}$ are independent for different $r$ and/or $k$ (due to sufficiently separated APs), the correlation matrix of $\boldsymbol{\mathfrak{h}}$ becomes the block diagonal matrix of the spatial correlation matrices $\overline{\mathbf{R}}_{r,k}$, i.e., 
\begin{align}\label{eq:R}
    &\normalfont\textbf{R}=\mathbb{E}\left\{\boldsymbol{\mathfrak{h}}\boldsymbol{\mathfrak{h}}^H\right\}\\
    &= \normalfont\mathrm{blkdiag}\left(\overline{\textbf{R}}_{1,1},\cdots,\overline{\textbf{R}}_{1,N_{\rm tx}},\cdots,\overline{\textbf{R}}_{N_{\rm rx},1},\cdots,\overline{\textbf{R}}_{N_{\rm rx},N_{\rm tx}} \right).\nonumber
\end{align}

\section{Sensing Performance Metrics}
Detection probability and false alarm probability are fundamental metrics used to assess sensing performance \cite{richards2010principles}. Detection probability, represented by $P_{\rm d}$, is the likelihood of correctly detecting a target when it is present. False alarm probability, denoted by $P_{\rm fa}$, is defined as the probability of incorrectly detecting a target when it is not present.

In addition to the abovementioned metrics, sensing mutual information (MI) has also been proposed in the literature to measure the sensing performance \cite{li2022framework,li2022integrated,yang2007mimo, ahmadipour2022information, xie2023sensing}. 
From an information-theoretic perspective, the objective of sensing is to retrieve environmental information embedded in the reflected signals from the target and the sensing MI is defined as the MI between the target's response and the reflected signals. It is shown that maximizing MI between the target's response and the reflected signals minimizes the MMSE of the estimation of target's response \cite{yang2007mimo}. Following the concept of sensing MI, sensing rate \cite{ouyang2022performance} 
or radar information rate \cite{meng2023network} are also defined as the sensing MI per unit time-frequency resource block. The authors of \cite{meng2023network} propose area spectral efficiency (ASE) as the network-level performance metric for both communication and sensing. 

However, the ASE expression in \cite{meng2023network} does not take into account the impact of the resource block usage in sensing. In ISAC approach, we employ $\tau_c-\tau_p$ symbols jointly for sensing and communication. This implies that we can have multiple sensing blocks in one coherence block given that $\tau_s<\tau_c$. To this end, we define the sensing SE as  
\begin{align}\label{eq: SE_S}
    \mathsf{SE}_{\rm s} = \frac{\tau_s^{\rm total}}{\tau_c} \log_2\left(1+\mathsf{SINR}_{\rm s}\right),
\end{align}
where $\tau_s^{\rm total}$ in the pre-log factor is the total number of symbols used for sensing in one coherence block and $\mathsf{SINR}_{\rm s}$ is the sensing SINR, which we will derive later.

\section{MAPRT Detector for Sensing}\label{section6}
In this paper, we employ the MAPRT detector for our system model and extend the derivation in \cite{guruacharya2020map}, which considers only a single transmitter and single-antenna receivers. 

Different from \cite{guruacharya2020map}, in our system, the transmitted signals are known due to the C-RAN architecture. Furthermore, the relationship between the received signals at the receive APs and the unknown RCSs is more complex due to the direction-dependent MIMO channels from the transmit APs to the receive APs through the target. Our derived test statistics demonstrates not only how the detector should be implemented but also how the different signals from the receive APs should be centrally fused in the edge cloud.      
The detection is applied using $\tau_s$ received sensing symbols for a particular target location.
 Let us define the vectors $\textbf{y}_{\tau}\in \mathbb{C}^{MN_{\rm rx}\tau_s}$, $\textbf{n}_{\tau}\in \mathbb{C}^{MN_{\rm rx}\tau_s}$, $\textbf{h}_{\tau}\in \mathbb{C}^{MN_{\rm rx}\tau_s}$, and $\textbf{g}_{\tau}\in \mathbb{C}^{MN_{\rm rx}\tau_s}$ constructed by concatenating $\textbf{y}[m]$, $\textbf{n}[m]$, $\textbf{X}[m]\boldsymbol{\mathfrak{h}}$, and $\textbf{G}[m]\boldsymbol{\alpha}$ in \eqref{y_m}, respectively for $\tau_s$ symbols. The binary hypothesis  used in the MAPRT detector is written as
\begin{align}\label{hypothesis}
   &\mathcal{H}_0 : \textbf{y}_{\tau}= \textbf{h}_{\tau}+\textbf{n}_{\tau} \nonumber\\
  &\mathcal{H}_1 :\textbf{y}_{\tau}=\textbf{g}_{\tau}+\textbf{h}_{\tau}+\textbf{n}_{\tau}.
\end{align}

 The null hypothesis $\mathcal{H}_0$ represents the case that there is no target in the sensing area, 
 while the alternative hypothesis $\mathcal{H}_1$ represents the existence of the target. 
We let $\mathcal{H}\in \{\mathcal{H}_0, \mathcal{H}_1\}$ be the set of hypotheses. We also recall that $\boldsymbol{\alpha}$ and  $\boldsymbol{\mathfrak{h}}$ are the vectors of unknown channel gains with respect to the RCS of the target and the target-free channels, respectively. The joint RCS estimation, target-free channel estimation, and hypothesis testing problem can be written as 
 \begin{equation} \label{joint testing problem}
    \left( \hat{\boldsymbol{\alpha}},\hat{\boldsymbol{\mathfrak{h}}},\hat{\mathcal{H}}\right) =\argmax_{\boldsymbol{\alpha}, \boldsymbol{\mathfrak{h}}, \mathcal{H}} p\left(\boldsymbol{\alpha}, \boldsymbol{\mathfrak{h}},\mathcal{H}|\textbf{y}_{\tau}\right), 
\end{equation}
where $p\left(\boldsymbol{\alpha},\boldsymbol{\mathfrak{h}},\mathcal{H}|\textbf{y}_{\tau}\right)$ is the joint probability density function (PDF) of $\boldsymbol{\alpha}$, $\boldsymbol{\mathfrak{h}}$, and $\mathcal{H}$ given the received signal vector $\textbf{y}_{\tau}$. 
Following the approach in \cite{guruacharya2020map} and \cite{behdad2022power} for our problem and assuming $\boldsymbol{\alpha}$ and $\boldsymbol{\mathfrak{h}}$ are independent of each other,  the corresponding MAPRT detector can be expressed as 
\begin{align} \label{eq:likelihood-MAPRT}
    \Lambda = \frac{\max\limits_{\boldsymbol{\alpha},\boldsymbol{\mathfrak{h}}} \, p\left(\textbf{y}_{\tau}|\boldsymbol{\alpha},\boldsymbol{\mathfrak{h}},\mathcal{H}_1\right) p\left(\boldsymbol{\alpha}|\mathcal{H}_1\right)p\left(\boldsymbol{\mathfrak{h}}|\mathcal{H}_1\right)}{\max\limits_{\boldsymbol{\mathfrak{h}}} \, p\left(\textbf{y}_{\tau}|\boldsymbol{\mathfrak{h}},\mathcal{H}_0\right)p\left(\boldsymbol{\mathfrak{h}}|\mathcal{H}_0\right) }\begin{matrix}
    \mathcal{H}_1\\\geq\\<\\ \mathcal{H}_0
    \end{matrix} \lambda_d,
\end{align}
 where $\lambda_d$ is the threshold used by the detector, which is selected to achieve a desired false alarm probability and $p(\cdot)$ denotes the PDF. 
In \eqref{eq:likelihood-MAPRT}, we maximize the joint PDFs with respect to the estimation parameters in \eqref{joint testing problem} under the hypothesis $\mathcal{H}_1$ and $\mathcal{H}_0$ in the numerator and the denominator, respectively. 
To streamline the analysis, we define the above-mentioned optimization problems by $\Lambda_{\rm num}$ and $\Lambda_{\rm den}$, respectively, which are expressed as
\begin{align}
    &\Lambda_{\rm num}= \max\limits_{\boldsymbol{\alpha},\boldsymbol{\mathfrak{h}}} \, p\left(\textbf{y}_{\tau}|\boldsymbol{\alpha},\boldsymbol{\mathfrak{h}},\mathcal{H}_1\right) p\left(\boldsymbol{\alpha}|\mathcal{H}_1\right)p\left(\boldsymbol{\mathfrak{h}}|\mathcal{H}_1\right), \label{maxUnderH1}\\
    &\Lambda_{\rm den}= \max\limits_{\boldsymbol{\mathfrak{h}}} \, p\left(\textbf{y}_{\tau}|\boldsymbol{\mathfrak{h}},\mathcal{H}_0\right)p\left(\boldsymbol{\mathfrak{h}}|\mathcal{H}_0\right)\label{maxUnderH0}.
\end{align}
The solutions to these optimization problems under $\mathcal{H}_0$ and $\mathcal{H}_1$ 
are presented in Lemma \ref{lemma:H0}.
\begin{lemma} \label{lemma:H0}
Consider the binary hypothesis defined in \eqref{hypothesis} and the overall received signal for sensing for all time instants $m \in \{1, \ldots,\tau_s\}$ as defined in \eqref{y_m}. Given that $\normalfont\bm{\mathfrak{h}} \sim \mathcal{CN}\left( \textbf{0}, \textbf{R}\right)$ and $\normalfont\boldsymbol{\alpha}\sim \mathcal{CN}\left(\textbf{0}, \textbf{R}_{\rm rcs}\right)$, where  $\normalfont\textbf{R}_{\rm rcs}$ is the covariance matrix of the correlated RCS values, we have the following results.

{\bf Optimization under $\mathcal{H}_0$:}
Given that $\mathcal{H}_0$ is true, 
the denominator of the likelihood ratio in \eqref{eq:likelihood-MAPRT} is
\begin{align} \label{Lambda_den}
  \Lambda_{\rm den}=  C_1 C_2 \exp\left(\normalfont\frac{\textbf{b}^H\textbf{D}^{-1}\textbf{b}-F}{\sigma_n^2}\right),
\end{align}
 where $C_1 = \frac{1}{\left(\pi \sigma_n^2\right)^{ MN_{\rm rx}\tau_s}}$ and $C_2 =\frac{1}{\pi^{M^2N_{\rm tx}N_{\rm rx}
    }\det\left(\normalfont\textbf{R}\right)}$.  The matrix $\normalfont\textbf{D}$, the vector $\normalfont\textbf{b}$, and the scalar $F$ are defined in \eqref{eq:C-D}, \eqref{eq:a-b}, and \eqref{eq:e-F}, respectively.
    
{\bf Optimization under $\mathcal{H}_1$:} Given that $\mathcal{H}_1$ is true, 
the numerator of the likelihood ratio in \eqref{eq:likelihood-MAPRT} is given as
\begin{align}\label{Lambda_num}
    &\Lambda_{\rm num}= C_1 C_2 C_3 \nonumber\\
    &\hspace{3mm} \times
    \exp\left(\frac{1}{\sigma_n^2} \normalfont\begin{bmatrix}
\textbf{a}^H &
\textbf{b}^H
\end{bmatrix} \begin{bmatrix}
 \textbf{C}& \textbf{E} \\ 
 \textbf{E}^H&\textbf{D} 
\end{bmatrix}^{-1}\begin{bmatrix}
\textbf{a}\\ 
\textbf{b}
\end{bmatrix}-\frac{1}{\sigma_n^2} F\right), 
\end{align}
where
\begin{align}
    &\normalfont\textbf{a}= \sum_{m=1}^{\tau_s} \textbf{G}^H[m] \textbf{y}[m],\quad \normalfont\textbf{b}= \sum_{m=1}^{\tau_s} \textbf{X}^H[m] \textbf{y}[m],\label{eq:a-b}\\
    &\normalfont\textbf{C}= \sum_{m=1}^{\tau_s} \textbf{G}^H[m]\textbf{G}[m]+ \sigma_n^2 \textbf{R}_{\rm rcs}^{-1}, \\
    &\normalfont\textbf{D}= \sum_{m=1}^{\tau_s} \textbf{X}^H[m]\textbf{X}[m]+ \sigma_n^2\textbf{R}^{-1}, \label{eq:C-D}\\
   & \normalfont\textbf{E}= \sum_{m=1}^{\tau_s} \textbf{G}^H[m]\textbf{X}[m],\quad \normalfont F = \sum_{m=1}^{\tau_s} \Vert \textbf{y}[m]\Vert^2. \label{eq:e-F}
\end{align}
\end{lemma}
\begin{proof}
See Appendix~\ref{appendix:A}.
\end{proof}

By substituting \eqref{Lambda_den} and \eqref{Lambda_num} into \eqref{eq:likelihood-MAPRT}, the MAPRT detector is obtained as \eqref{eq:lambda}. 
\begin{figure*}[t!]
\begin{align}\label{eq:lambda}
    \Lambda = C_3 \exp\left(\frac{1}{\sigma_n^2}\begin{bmatrix}
\textbf{a}^H &
\textbf{b}^H
\end{bmatrix}\left( \begin{bmatrix}
 \textbf{C}& \textbf{E} \\ 
 \textbf{E}^H&\textbf{D} 
\end{bmatrix}^{-1}-\begin{bmatrix}
 \textbf{0}& \textbf{0} \\ 
 \textbf{0}&\textbf{D}^{-1} 
\end{bmatrix}\right)\begin{bmatrix}
\textbf{a}\\ 
\textbf{b}
\end{bmatrix}\right)\begin{matrix}
    \mathcal{H}_1\\[-1mm]\geq\\[-.5mm] <\\[-1mm] \mathcal{H}_0
    \end{matrix} \lambda_d.
\end{align}
\hrulefill
\end{figure*}
We define the test statistics $T \triangleq \ln(\Lambda)$, which is given as
\begin{align}
    T &= \frac{1}{\sigma_n^2}\left(\begin{bmatrix}
\textbf{a}^H &
\textbf{b}^H
\end{bmatrix}\left( \begin{bmatrix}
 \textbf{C}& \textbf{E} \\ 
 \textbf{E}^H&\textbf{D} 
\end{bmatrix}^{-1}-\begin{bmatrix}
 \textbf{0}& \textbf{0} \\ 
 \textbf{0}&\textbf{D}^{-1} 
\end{bmatrix}\right)\begin{bmatrix}
\textbf{a}\\ 
\textbf{b}
\end{bmatrix}\right)\nonumber\\
&+\ln(C_3).
    \end{align}
Finally, the decision is taken as follows
\begin{align}
    \hat{\mathcal{H}}= \left\{\begin{matrix}
\mathcal{H}_1 & \textrm{if}& T\geq \ln(\lambda_d),
\\ 
\mathcal{H}_0 &\textrm{if}& T <\ln(\lambda_d).
\end{matrix}\right. \label{detection}
\end{align}

\section{Power Allocation Algorithms}\label{section7}
In Section \ref{section6}, we have derived the MAPRT detector. It uses the already selected power allocation strategy. 
In this section, we first propose the power allocation algorithms for the case of integrated sensing and communication where we utilize communication symbols and optionally dedicated sensing symbols for target detection. Next, we describe two benchmarks algorithms: communication-centric and orthogonal sharing. The former represents a power allocation algorithm that only considers communication constraints and aim to minimize the total power consumption without taking into account the sensing requirement. The latter corresponds to the case when the orthogonal resources are shared orthogonally between sensing and communication.  

\subsection{Integrated Sensing and Communication Algorithms}
In this section, we maximize the sensing SINR, denoted by $\mathsf{SINR}_{\rm s}^{\rm isac}$, under the presumption that the target is present. 
The optimization problem can be cast as
\begin{subequations}\label{optimization0}
\begin{align} \label{optimization}
   \textbf{P1:}\quad   \underset{\boldsymbol{\rho}\geq \textbf{0}}{\textrm{maximize}} \quad &\mathsf{SINR}_{\rm s}^{\rm isac} \\ 
    \textrm{subject to} \quad 
   &\mathsf{SINR}^{(\rm dl)}_{i}  \geq \gamma_c, \quad i=1, \ldots, N_{\rm ue}\\
  & P_k \leq P_{\rm tx},\quad k=1,\ldots ,N_{\rm tx} 
\end{align}
\end{subequations}
where $\gamma_c$ is the minimum required SINR threshold to provide a certain SE to all the UEs and $P_{\rm tx}$ is the maximum transmit power per AP. 
 The sensing SINR is given as 
\begin{align} \label{gamma_s_main}
     \mathsf{SINR}_{\rm s}^{\rm isac}=\frac{\boldsymbol{\rho}^T \textbf{A} \boldsymbol{\rho}}{\tau_s MN_{\rm rx}\sigma_n^2+ \boldsymbol{\rho}^T\textbf{B}\boldsymbol{\rho} },
\end{align}
where
\begin{align}
    \textbf{A}=& \sum_{m=1}^{\tau_s}\textbf{D}_{\rm s}^H[m] \Bigg(\sum_{r=1}^{N_{\rm rx}}\sum_{k=1}^{N_{\rm tx}}\sum_{j=1}^{N_{\rm tx}}\sqrt{\beta_{r,k}\beta_{r,j}}\textbf{W}_k^H\textbf{a}^{*}(\varphi_{k},\vartheta_{k})\textbf{a}^H(\phi_{r},\theta_{r})\nonumber\\
 &\hspace{3mm}\times\textrm{cov}\left( \alpha_{r,j},\alpha_{r,k}\right)\textbf{a}(\phi_{r},\theta_{r})\textbf{a}^{T}(\varphi_{j},\vartheta_{j})  \textbf{W}_j\Bigg)\textbf{D}_{\rm s}[m],
\end{align}
and
\begin{align}\textbf{B}&=\sum_{m=1}^{\tau_s}\textbf{D}_{\rm s}^H[m] \left(\sum_{r=1}^{N_{\rm rx}}\sum_{k=1}^{N_{\rm tx}}\mathrm{tr}\left(\textbf{R}_{\mathrm{rx},(r,k)}\right)\textbf{W}_k^H \textbf{R}^{T}_{{\rm tx},(r,k)}\textbf{W}_k \right)\textbf{D}_{\rm s}[m].
\end{align}

The derivation of \eqref{gamma_s_main} is given in Appendix~\ref{app:proof_SNR}.
The matrices $\textbf{A}$ and $\textbf{B}$ are Hermitian symmetric. Since the entries of $\boldsymbol{\rho}$ are real valued, we have $\boldsymbol{\rho}^T\textbf{A}\boldsymbol{\rho}=\boldsymbol{\rho}^T\Re(\textbf{A})\boldsymbol{\rho}$ and $\boldsymbol{\rho}^T\textbf{B}\boldsymbol{\rho}=\boldsymbol{\rho}^T\Re(\textbf{B})\boldsymbol{\rho}$. Let us define $\textbf{A}_{r}=\Re(\textbf{A})$ and $\textbf{B}_{r}=\Re(\textbf{B})$. Then, the sensing SINR can be written as
\begin{align}
     \mathsf{SINR}_{\rm s}^{\rm isac}=\frac{\boldsymbol{\rho}^T \textbf{A}_r \boldsymbol{\rho}}{\tau_s MN_{\rm rx}\sigma_n^2+ \boldsymbol{\rho}^T\textbf{B}_r\boldsymbol{\rho} }.
\end{align}

By following an approach from \cite{benson2006fractional}, it can be shown that the optimization problem in \eqref{optimization0} is equivalent to 
 \begin{subequations}
\begin{align} 
    \underset{\boldsymbol{\rho}\geq \textbf{0}, \ t}{\textrm{maximize}} \quad &\frac{\boldsymbol{\rho}^T \textbf{A}_r \boldsymbol{\rho}}{t} \\
    \textrm{subject to} \quad 
    &\tau_s MN_{\rm rx}\sigma_n^2+ \boldsymbol{\rho}^T\textbf{B}_r\boldsymbol{\rho} \leq t \label{cona}\\
    &\mathsf{SINR}^{(\rm dl)}_{i}  \geq \gamma_c, \quad i=1, \ldots, N_{\rm ue}\label{conb}\\
  & P_k \leq P_{\rm tx},\quad k=1,\ldots ,N_{\rm tx}. \label{conc}
\end{align}
\end{subequations}

Let us define the objective function as $f\left(\boldsymbol{\rho},t\right) = \frac{\boldsymbol{\rho}^T \textbf{A}_r \boldsymbol{\rho}}{t}$ in terms of the  optimization variable $\begin{bmatrix}
 \boldsymbol{\rho}^T & t
 \end{bmatrix}^T$. The function $f\left(\boldsymbol{\rho},t\right)$ is a convex function. The proof is given in Appendix~C.
 
Utilizing the convexity of the objective function and utilizing the first-order Taylor approximation around the previous iterate $\left(\boldsymbol{\rho}^{(c-1)}, t^{(c-1)}\right)$, where $c$ is the iteration counter, the objective function is lower bounded as 
\begin{align}
   f\left(\boldsymbol{\rho}, t\right)&\geq f\left(\boldsymbol{\rho}^{(c-1)}, t^{(c-1)}\right)+ \left(\begin{bmatrix}
 \boldsymbol{\rho} \\ t
 \end{bmatrix}-\begin{bmatrix}
 \boldsymbol{\rho}^{(c-1)} \\ t^{(c-1)}
 \end{bmatrix}\right)^T\nonumber\\
 &\nabla f\left ( \boldsymbol{\rho}^{(c-1)},t^{(c-1)} \right ).
 \end{align}
 At iteration $c$, instead of the actual objective function, we maximize its lower bound with the concave-convex procedure. After removing the constant terms that are independent of the optimization variables in use, we obtain the linearized objective function to be maximized as
 \begin{align}
 &\begin{bmatrix}
 \boldsymbol{\rho} \\ t
 \end{bmatrix}^T \nabla f\left ( \boldsymbol{\rho}^{(c-1)},t^{(c-1)} \right )= \begin{bmatrix}
 \boldsymbol{\rho} \\ t
 \end{bmatrix}^T \begin{bmatrix}
     \frac{2\textbf{A}_r\boldsymbol{\rho}^{(c-1)}}{t^{(c-1)}} \\ -\frac{\left(\boldsymbol{\rho}^{(c-1)}\right)^T\textbf{A}_r\boldsymbol{\rho}^{(c-1)}}{\left(t^{(c-1)}\right)^2}
     \end{bmatrix} \nonumber\\
 &\hspace{10mm}=  \frac{2\boldsymbol{\rho}^T \textbf{A}_r \boldsymbol{\rho}^{(c-1)}}{t^{(c-1)}} - \frac{\left(\boldsymbol{\rho}^{(c-1)}\right)^T\textbf{A}_r\boldsymbol{\rho}^{(c-1)}}{\left(t^{(c-1)}\right)^2} t\nonumber\\
 & \hspace{10mm}= \left(2\boldsymbol{\rho}- \frac{t}{t^{(c-1)}}\boldsymbol{\rho}^{(c-1)}\right)^T \left(\frac{\textbf{A}_r \boldsymbol{\rho}^{(c-1)}}{t^{(c-1)}} \right).
\end{align}
 Now, the optimization problem can be rewritten as
 \begin{subequations} \label{eq:optimization1}
\begin{align} 
    \underset{\boldsymbol{\rho}\geq \textbf{0}, \  t}{\textrm{minimize}} \quad &-\left(2\boldsymbol{\rho}- \frac{t}{t^{(c-1)}}\boldsymbol{\rho}^{(c-1)}\right)^T \left(\frac{\textbf{A}_r \boldsymbol{\rho}^{(c-1)}}{t^{(c-1)}} \right) \\
    \textrm{subject to} \quad& \eqref{cona},\eqref{conb},\eqref{conc}.
\end{align}
\end{subequations}

The quadratic constraints in \eqref{cona} are convex  and the communication SINR constraints in \eqref{conb} can be rewritten as second-order cone (SOC) constraints in terms of $\boldsymbol{\rho}=[\sqrt{\rho_0} \ \ldots \sqrt{\rho_{N_{\rm ue}}}]^T$ as
\begin{align}
    &\left \Vert \begin{bmatrix}  a_{i,0}\sqrt{\rho_{0}} & a_{i,1}\sqrt{\rho_{1}} & \!\ldots\! &
    a_{i,i}\sqrt{\rho_{i}} & \!\ldots\! & a_{i,N_{\rm ue}}\sqrt{\rho_{N_{\rm ue}}} &
    \sigma_n
    \end{bmatrix} \right \Vert \nonumber\\
    &\leq \frac{\sqrt{\rho_i} b_{i}}{\sqrt{\gamma_c}}, \quad i=1, \ldots, N_{\rm ue}, \label{conb2}
\end{align}
where 
\begin{align}
& b_i = \left \vert \mathbb{E}\left\{\textbf{h}_{i}^H \textbf{w}_{i}\right\}\right\vert,\quad i= 1,\ldots,N_{\rm ue}\\
&a_{i,0} = \sqrt{\mathbb{E}\left\{\left \vert\tilde{\textbf{h}}_{i}^H \textbf{w}_0\right\vert^2\right\}},\quad i=1, \ldots,N_{\rm ue} \\
& a_{i,j} =\sqrt{\mathbb{E}\left\{\left \vert\textbf{h}_{i}^H \textbf{w}_{j}\right\vert^2\right\}},\quad i,j=1, \ldots,N_{\rm ue}, \quad j\neq i\\
& a_{i,i} =\sqrt{\mathbb{E}\left\{\left \vert\textbf{h}_{i}^H \textbf{w}_{i}\right\vert^2\right\}- b_i^2}, \quad i= 1,\ldots,N_{\rm ue}
\end{align}
from \eqref{sinr_i}. The third set of constraints in \eqref{conc} can also be rewritten in SOC form as
\begin{align}
\left \Vert \textbf{F}_k\boldsymbol{\rho}\right \Vert \leq \sqrt{P_{\rm tx}}, \quad k=1,\ldots,N_{\rm tx} \label{conc2}
\end{align}
where $\textbf{F}_k =\textrm{diag}\left( \sqrt{\mathbb{E}\left\{\Vert\textbf{w}_{0,k}\Vert^2\right\}}, \ldots,\sqrt{\mathbb{E}\left\{\Vert \textbf{w}_{N_{\rm ue},k}\Vert^2\right\}}\right)$ from \eqref{eq:Pk}.
Now, the optimization problem in \eqref{eq:optimization1} can be expressed as a concave-convex programming problem as  
\begin{subequations}\label{optimization2}
\begin{align} 
   &\underset{\boldsymbol{\rho}\geq \textbf{0}, \ t}{\textrm{minimize}} \quad  -\left(2\boldsymbol{\rho}- \frac{t}{t^{(c-1)}}\boldsymbol{\rho}^{(c-1)}\right)^T \left(\frac{\textbf{A}_r \boldsymbol{\rho}^{(c-1)}}{t^{(c-1)}} \right) \\
 &   \textrm{subject to} \quad  \text{\eqref{cona}, \eqref{conb2}, \eqref{conc2}}.
  \end{align}
  \end{subequations}
 
 This problem can be solved using the concave-convex procedure, whose steps are outlined in Algorithm~\ref{alg:fixed-point-uplink}. Under some mild conditions, this algorithm is guaranteed to converge to a stationary point of the problem \cite{lanckriet2009convergence}.

\begin{algorithm}[t]
	\caption{Concave-Convex Procedure for Power Allocation} \label{alg:fixed-point-uplink}
	\begin{algorithmic}[1]
		\State {\bf Initialization:} Set an arbitrary initial positive $\boldsymbol{\rho}^{(0)}$ and the solution accuracy $\epsilon>0$. Set the iteration counter to $c=0$ and $t^{(0)} =\tau_s MN_{\rm rx}\sigma_n^2+ \left(\boldsymbol{\rho}^{(0)}\right)^T\textbf{B}_r\boldsymbol{\rho}^{(0)}  $.
		 \State $c\leftarrow c+1$.
		\State Set $\boldsymbol{\rho}^{(c)}$ and $t^{(c)}$ to the solution of the convex problem in \eqref{optimization2},  where the previous iterate $\boldsymbol{ \rho}^{(c-1)}$ and $t^{(c-1)}$ are taken as constant.
        \State If $\left(2 \left(\boldsymbol{\rho^{(c)}-\rho^{(c-1)}}\right)-\frac{t^{(c)}-t^{(c-1)}}{t^{(c-1)}} \boldsymbol{\rho^{(c-1)}}\right)^T \textbf{A}_r\frac{\boldsymbol{\rho}^{(c-1)}}{t^{(c-1)}} \leq \epsilon $, terminate the iterations. Otherwise return to Step 2.
		\State {\bf Output:} The transmit power coefficients  $\boldsymbol{\rho}^{(c)}$.
	    \end{algorithmic}
\end{algorithm}
\subsection{Communication-Centric Algorithm}\label{sec:CC algo}
In this algorithm, sensing is performed simultaneously with the downlink communication but the power allocation is optimized to minimize the transmit power consumption taking into account only communication constraints. The optimization problem can be cast as
\begin{subequations}\label{optimization_P2}
\begin{align} \label{optimization}
 \textbf{P2:}\quad   \underset{\boldsymbol{\rho}\geq \textbf{0}}{\textrm{minimize}} \quad & P_{\rm total}=\sum_{k=1}^{N_{\rm tx}}P_k \\ 
    \textrm{subject to} \quad &\mathsf{SINR}^{(\rm dl)}_{i}  \geq \gamma_c, \quad i=1, \ldots, N_{\rm ue}\\
  & P_k \leq P_{\rm tx},\quad k=1,\ldots ,N_{\rm tx}. 
\end{align}
\end{subequations}
\subsection{Orthogonal Sharing Algorithm}\label{sec: OS algo}
In this algorithm, sensing and communication are conducted separately where we divide the resources (symbols) between sensing and communication in an orthogonal manner over one coherence block.
Since we have separate sensing and communication, the sensing and communication power coefficients should be optimized separately. In this case, the optimization problem for communication would be the same as \eqref{optimization_P2}. But the achievable SE of UE $i$ becomes
\begin{align} \label{eq:SE_i_os}
    \mathsf{SE}^{(\rm dl,os)}_i = \frac{\tau_c-\tau_p-\tau_s}{\tau_c} \log_2\left(1+\mathsf{SINR}^{(\rm dl)}_i\right),
\end{align}
where $\tau_s \leq \tau_c-\tau_p$.
The optimization problem for sensing can be written as
\begin{subequations}\label{optimization_P1}
\begin{align} \label{optimization}
 \textbf{P3:}\quad   \underset{\rho_{0,k}\geq 0}{\textrm{maximize}} \quad &\mathsf{SINR}_{\rm s}^{\rm os}\\ 
    \textrm{subject to} \quad 
  & \rho_{0,k}\mathbb{E}\left\{\Vert \textbf{w}_{0,k} \Vert^2\right\} \leq P_{\rm tx},\quad \forall k
\end{align}
\end{subequations}
which aims to maximize the sensing SINR while the power consumption per AP does not exceed $P_{\rm tx}$. 

In this scenario, communication symbols are not contributing to the sensing performance. Therefore, the received signal for sensing at receive AP $r$ is 
\begin{align}\label{y_rPrim_benchmark}
          \textbf{y}_r[m] =& \sum_{k=1}^{N_{\rm tx}} \alpha_{r,k}  \sqrt{\beta_{r,k}}\textbf{a}(\phi_{r},\theta_{r})\textbf{a}^{T}(\varphi_{k},\vartheta_{k})\sqrt{\rho_{0,k}}\textbf{w}_{0,k}s_0[m]\nonumber\\
          &+\sum_{k=1}^{N_{\rm tx}}  \textbf{H}_{r,k}\sqrt{\rho_{0,k}}\textbf{w}_{0,k}s_0[m]+\textbf{n}_r[m],
\end{align}
and sensing SINR is obtained as
\begin{align} \label{gamma_s_main_bechmark}
     \mathsf{SINR}_{\rm s}^{\rm os}=\frac{\sum_{k=1}^{N_{\rm tx}}\sum_{j=1}^{N_{\rm tx}}\sqrt{\rho_{0,k}}\, a_{k,j}\sqrt{\rho_{0,j}}}{\tau_s MN_{\rm rx}\sigma_n^2+ \sum_{k=1}^{N_{\rm tx}}\rho_{0,k}\, b_k},
\end{align}
where 
\begin{align}
    &a_{k,j}= \sum_{m=1}^{\tau_s} \vert s_0[m] \vert^2 \Bigg(\sum_{r=1}^{N_{\rm rx}}\sqrt{\beta_{r,k}\beta_{r,j}}\textrm{cov}\left( \alpha_{r,j},\alpha_{r,k}\right)\nonumber\\
 &\hspace{3mm}\times\textbf{w}_{0,k}^H\textbf{a}^{*}(\varphi_{k},\vartheta_{k})\textbf{a}^H(\phi_{r},\theta_{r})\textbf{a}(\phi_{r},\theta_{r})\textbf{a}^{T}(\varphi_{j},\vartheta_{j})  \textbf{w}_{0,j}\Bigg) ,
\end{align}
and
\begin{align}
b_{k}&=\sum_{m=1}^{\tau_s} \vert s_0[m] \vert ^2\left(\sum_{r=1}^{N_{\rm rx}}\mathrm{tr}\left(\textbf{R}_{\mathrm{rx},(r,k)}\right)\textbf{w}_{0,k}^H \textbf{R}^{T}_{{\rm tx},(r,k)}\textbf{w}_{0,k} \right).
\end{align}

 \section{Numerical Results}\label{section8}
In this section, we provide numerical results to quantify the trade-off between communication and sensing using the proposed MAPRT detector and power allocation algorithm. A total area of 500\,m $\times$ 500\,m is considered within which UEs are randomly generated. In each sensing block, the target's location is fixed and known, but it can be anywhere within a 15\,m $\times$ 15\,m sensing hotspot region located in the center of the area. The interference $\rho_0\normalfont\mathbb{E}\{\vert \tilde{\textbf{h}}_i^H \textbf{w}_0\vert^2\}$ created by the sensing beams to the UEs in \eqref{sinr_i} is computed by averaging  over 1000\, target locations in this region. We assume that the $N_{\rm tx}+N_{\rm rx}$ APs are also uniformly generated in the network area, where $N_{\rm tx}=16$ and the $N_{\rm rx}=2$ closest APs to the sensing hotspot area are selected as the sensing receivers. All APs are fixed during the simulation. The 2D locations of all the APs and the target are illustrated in Fig~\ref{locations}. We consider a total of $N_{\rm ue}=8$ UEs.

The maximum transmit power per AP is $1$\,W and the uplink pilot transmission power is $0.2$\,W.
The path loss for the communication channels and the target-free channels are modeled by the 3GPP Urban Microcell model defined in \cite[Table B.1.2.1-1]{3gpp2010further} with the difference that for the latter the channel gains are multiplied by an additional scaling parameter, $\mathfrak{s}\in (0,1)$, to suppress the known parts of the target-free channels due to LOS and permanent obstacles. Hence, the higher the value of $\mathfrak{s}$ is, the stronger the clutter power is. On the other hand, the sensing channel gains are computed by \eqref{eq:radar_range_equation}. The carrier frequency, the bandwidth, and the noise variance are set to $1.9$\,GHz, $20$\,MHz, and $-94$\,dBm, respectively.   

The spatial correlation matrices for the communication channels are generated by using the local scattering model in \cite[Sec.~2.5.3]{cell-free-book}, while the normalized RCS of the target is modeled by the Swerling-I model where $\alpha_{r,k} \sim \mathcal{CN}(0,1)$. For simplicity, we assume that the RCS values are independent and all RCS variances are the same, i.e., $\sigma_{r,k}^2= \sigma_{\rm rcs}^2$ for each pair of transmit AP $k$ and receive AP $r$.

The communication SINR threshold is 3\,dB for each UE.
The detection threshold, $\lambda_d$ in \eqref{detection}, is determined based on the false alarm probability $P_{\rm fa}$ of either $0.1$ or $0.01$, which are relevant for radar applications \cite{guruacharya2020map} and can be improved by using specific techniques \cite{richards2010principles}.  We use the MAPRT detector over one sensing block of $\tau_s$ symbols. Each coherence block contains $\tau_c= 200$ symbols with  $\tau_p = 10$ pilot symbols for uplink channel estimation. The sensing block, $\tau_s$, is set to $10$ and the scaling parameter $\mathfrak{s}$ is set to 0.3 unless otherwise stated. Note that the number of required samples used for sensing depends on several factors, such as the angular range and time interval for beam scanning in the sensing hotspot region \cite{richards2010principles}. When the sensing channel gain is lower due to the undesired atmospheric conditions or a smaller false alarm probability is required, more symbols would be needed to achieve the same detection probability. In this paper, the parameters are selected to cover a large range of RCS and detection probability values for a comprehensive comparison between several methods. 
\begin{figure}[tbp]
\centerline{\includegraphics[trim={2mm 0mm 9mm 2mm},clip,width=0.9\linewidth]{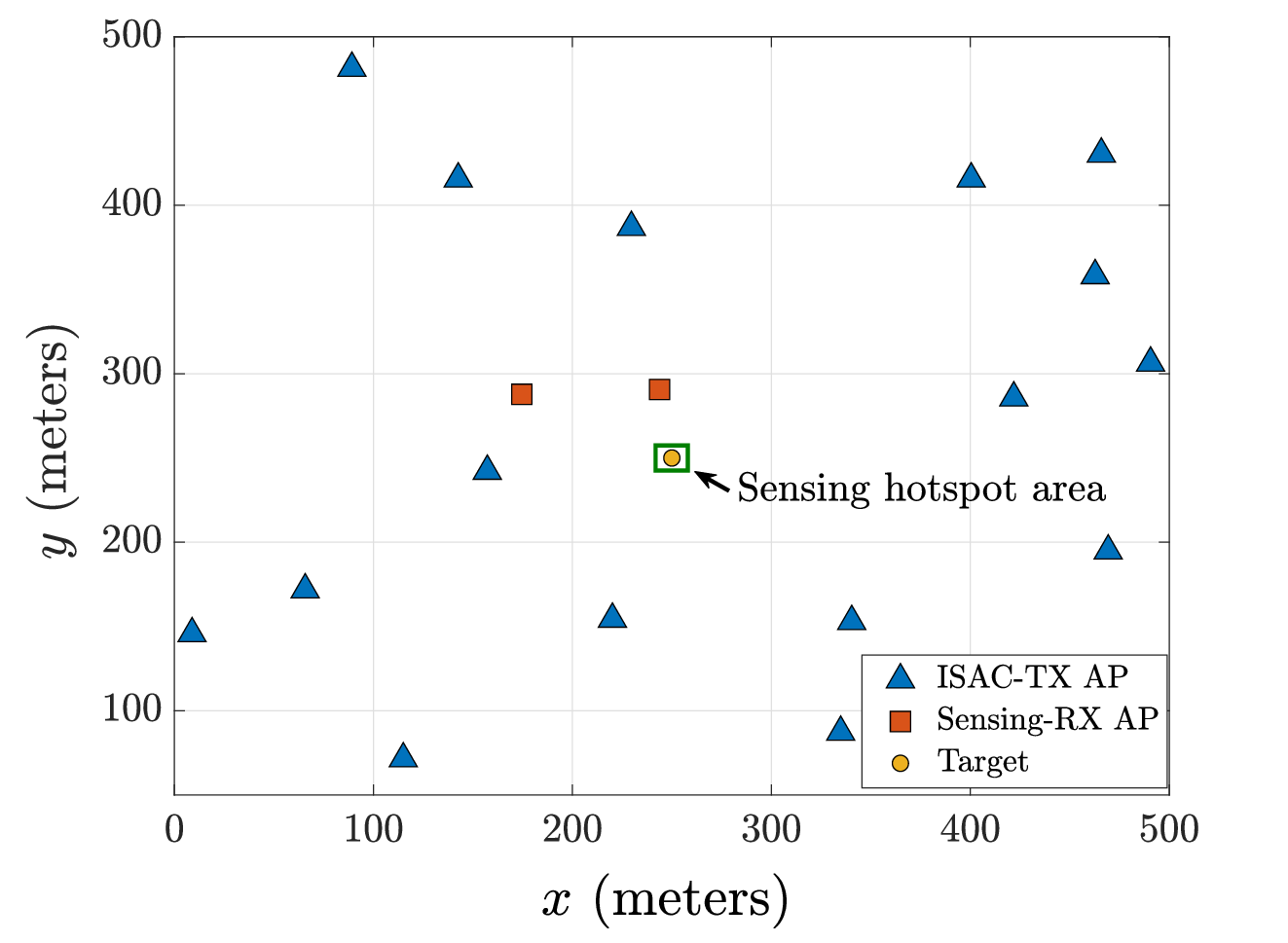}}
\caption{The 2D locations of the APs and the target.}\label{locations} 
\end{figure}

\begin{table}[]
\centering
	\caption{List of simulation parameters.}  \label{tab:simulation}
    
    \begin{tabular}{|c|c|c|c|}\hline
        $N_{\rm tx}, N_{\rm rx}, N_{\rm ue}$ & 16,\,2\,,8\,  & $P_{\rm ul}, P_{\rm tx}$ & 0.2\,W, 1\,W \\
        \hline
         $\tau_c, \tau_p$& 200,\,10 & $P_{\rm fa}$ & 0.1,\,0.01\\
         \hline
         $\gamma_c$&3\,dB  & $\sigma_n^2$ & -94\,dBm\\
         \hline
    \end{tabular}
    
\end{table}

The detection probability, $P_{\rm d}$, is computed by averaging over 100 random UE locations and channel realizations. For each set of UE locations, the expectations in \eqref{sinr_i} are obtained by taking average over 1000 random channel realizations. In each setup, $\frac{10^2}{P_{\rm fa}}$ random RCS and $\frac{10^2}{P_{\rm fa}}\times\tau_s$ random noise realizations are taken into consideration for accurate estimation of $P_{\rm d}$. The simulation parameters are listed in Table~\ref{tab:simulation}.

To evaluate the efficiency of the proposed power allocation algorithm, a communication-centric approach and an orthogonal sharing strategy, as described in Sections~\ref{sec:CC algo} and \ref{sec: OS algo}, are utilized. 
Therefore, four algorithms are compared in this section:  i) the  communication-centric (\textit{comm.-centric}); ii) the proposed ISAC power allocation algorithm ii.a) without additional sensing symbols (\textit{ISAC}) and ii.b) with dedicated additional sensing symbols (\textit{ISAC+S}); and iii) the orthogonal sharing algorithm (\textit{OS}). 

\subsection{Sensing-Communication SE Region}
Fig.~\ref{fig:SE} illustrates the sensing-communication SE region with the OS algorithm as well as those provided by the ISAC algorithms. The desired point is where both the sensing SE and communication SE are at their maximum available values, which is termed \textit{Utopia} in Fig.~\ref{fig:SE}. Notably, the OS algorithm exhibits a trade-off between the sensing and communication SE such that increasing one results in a decrease in the other. This is because the symbols in a coherence block are divided for the sensing and communication purposes, and thus the pre-log factors in \eqref{eq:SE_i_os} and \eqref{eq: SE_S} are changed. Conversely, the ISAC algorithms utilize all the available symbols for downlink (i.e., $\tau_c-\tau_p$) which results in the maximum communication SE under a given communication SINR threshold. It comes with the cost of a lower sensing SE than the \textit{Utopia} point because the sensing should share the transmit power with the communication. Nevertheless, both \textit{ISAC} and \textit{ISAC+S} place themselves beyond the sensing-communication SE region of the OS algorithm, demonstrating the effectiveness of the integrated operation of sensing and communication. Furthermore, \textit{ISAC+S} outperforms \textit{ISAC} algorithm by providing higher sensing SE while maintaining the maximum communication SE.

\begin{figure}[t!]

    \centering
    \includegraphics[trim={2mm 0mm 9mm 2mm},clip,width=0.9\linewidth]{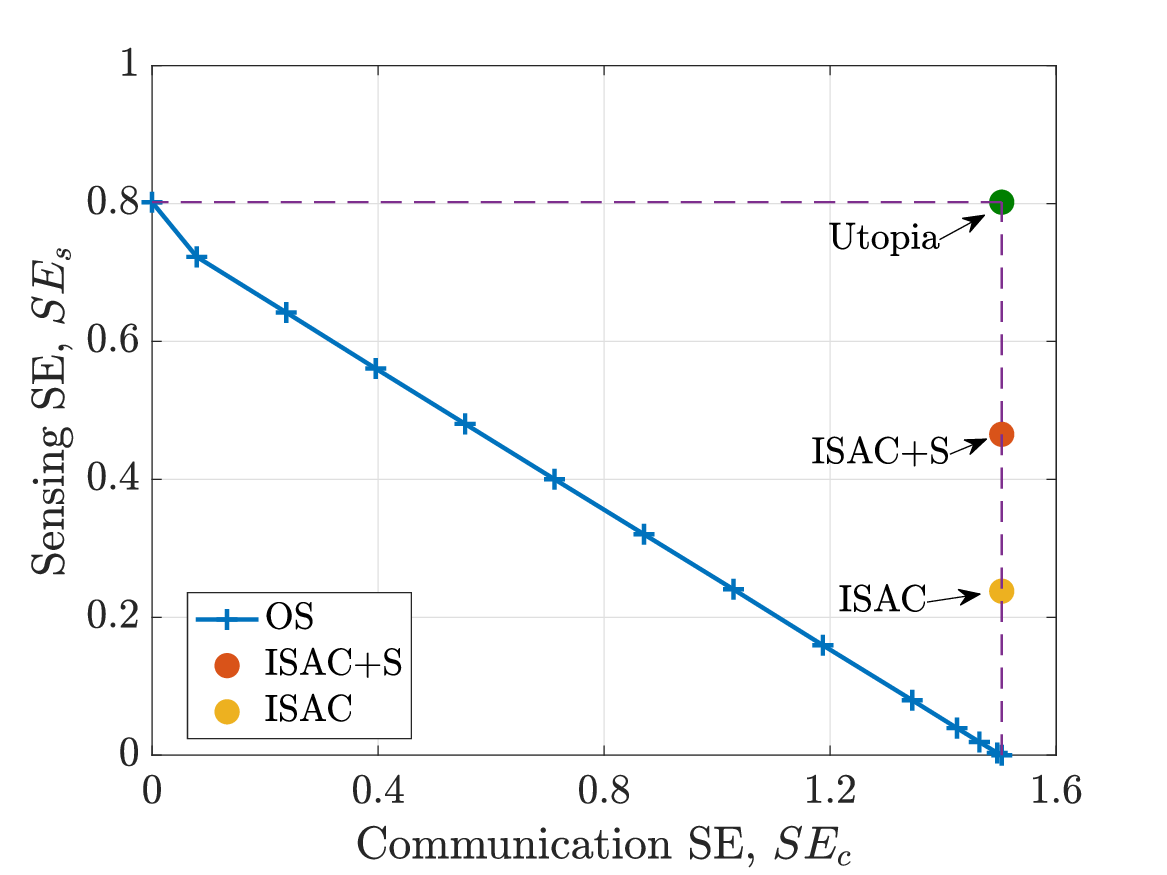}
    \caption{Sensing-communication SE region for $\sigma_{\rm rcs}^2 = -10$\,dBsm, $P_{\rm fa}=0.1$, and scaling parameter $\mathfrak{s}=0.01$.}
    \label{fig:SE}

\end{figure}

\subsection{Impact of the Sensing Blocklength}

\begin{figure}[t!]
		\centering
		\includegraphics[trim={2mm 0mm 9mm 4mm},clip,width=0.9\linewidth]{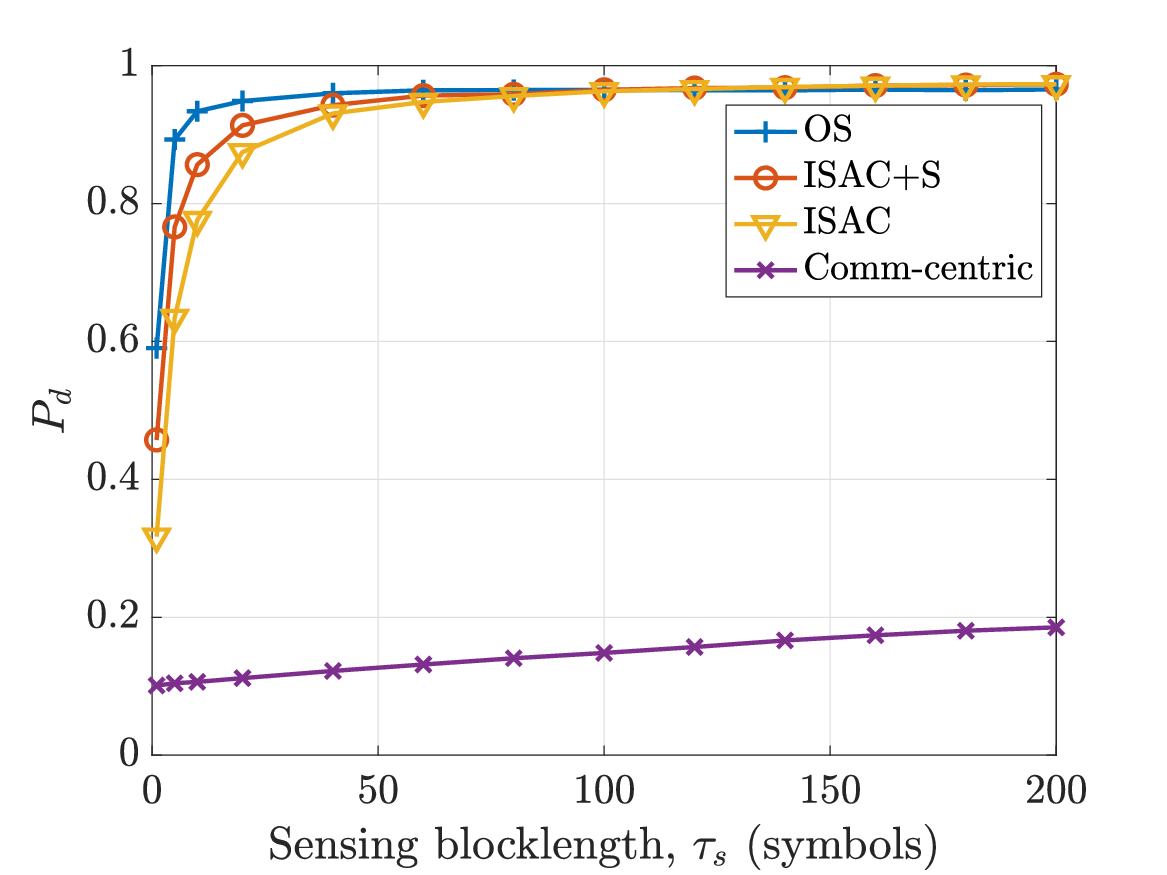}
			\caption{$P_{\rm d}$ vs. sensing blocklength $\tau_s$  for $\sigma_{\rm rcs}^2 = -10$\,dBsm, $P_{\rm fa}=0.1$, and scaling parameter $\mathfrak{s}=0.01$.} \label{Pd_sensingDuraion}
	\end{figure}
Fig.~\ref{Pd_sensingDuraion} shows the detection probability versus the sensing blocklength $\tau_s$. We observe that with all algorithms except \textit{comm-centric}, detection probability significantly grows by increasing the sensing blocklength while it becomes stable after $\tau_s = 80$ symbols. According to the results, the \textit{OS} algorithm outperforms the two other algorithms in low sensing blocklength values (i.e., $\tau_s <80$). This is because in \textit{OS} the power is only used to maximize the sensing SINR while in ISAC algorithms maximizing sensing SINR is limited due to communication constraints. However, the performance of the proposed ISAC algorithms without/with additional sensing beams (i.e., \textit{ISAC} and \textit{ISAC+S}) converges to the performance of \textit{OS} algorithm when $\tau_s \geq 80$ symbols.
 
\subsection{Impact of RCS Variance and False Alarm Probability Threshold}
\begin{figure}[t!]
	\begin{minipage}[t]{1\linewidth}
		\begin{center}
			\includegraphics[trim={2mm 0mm 9mm 4mm},clip,width=0.9\textwidth]{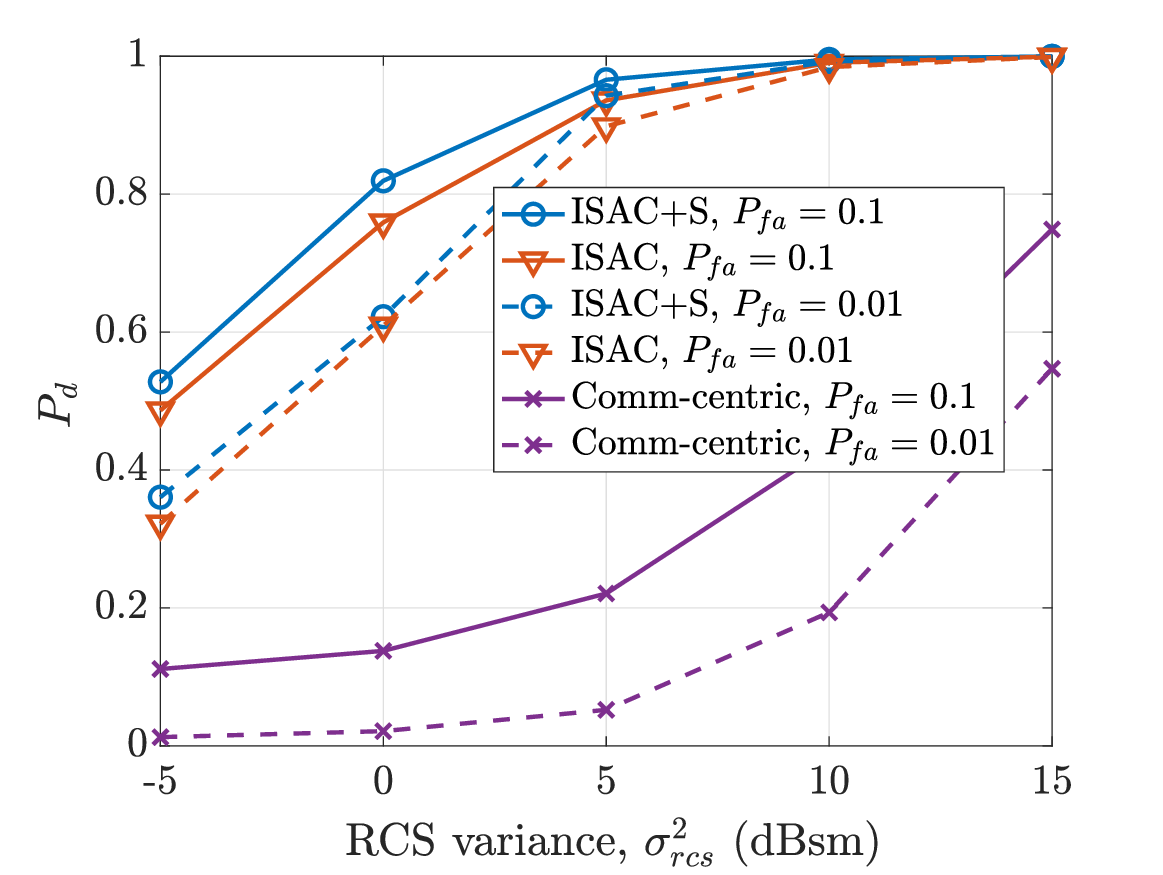}
			\caption{$P_{\rm d}$ vs. RCS variance for $\mathfrak{s}=0.3$, $P_{\rm fa}=0.1$,\,$0.01$.} \label{PD_RCS}
		\end{center}
	\end{minipage}
	\begin{minipage}[t]{1\linewidth}
		\begin{center}
			\includegraphics[trim={2mm 0mm 9mm 2mm},clip,width=0.9\textwidth]{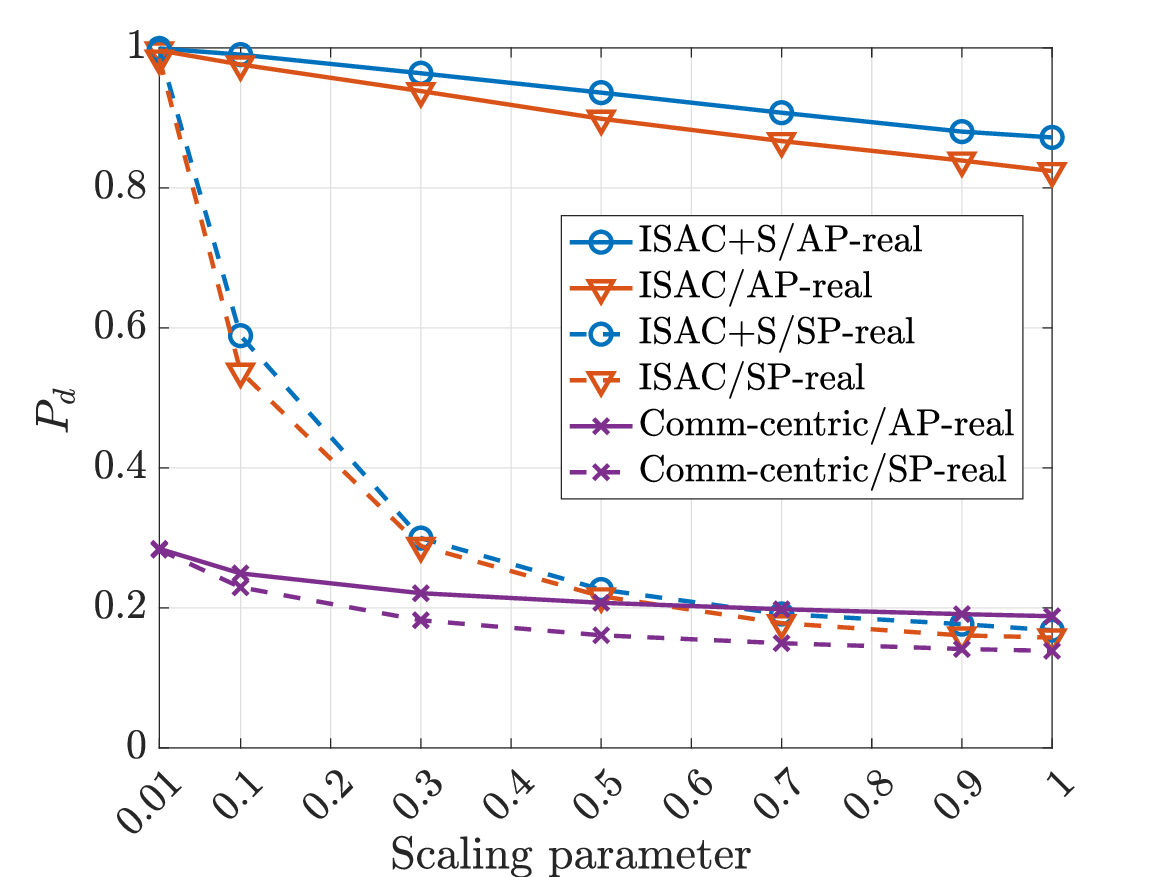}
			\caption{$P_{\rm d}$ vs. scaling parameter $\mathfrak{s}$ for
 $\sigma_{rcs}^2= 5$\,dB$\textrm{sm}$ and $P_{\rm fa}=0.1$.} \label{Pd_scalingParam}
		\end{center}
	\end{minipage} 
\end{figure}

Fig. \ref{PD_RCS} depicts the detection probability versus the variance of the RCS for false alarm probabilities of 0.1 and 0.01. The results suggest that allowing a higher false alarm probability threshold leads to increased detection probability. 
For both false alarm probabilities 0.1 and 0.01, the proposed \textit{ISAC} and \textit{ISAC+S} algorithms can achieve a detection probability close to 1 when the RCS variance is $\geq$ 10\,$\rm{dBsm}$ (``dBsm'' is the unit of RCS variance expressed in decibels relative to a square meter). In contrast, the \textit{comm-centric} algorithm fails to attain a $P_{\rm d}$ above 0.8 due to neglecting the sensing requirements. The detection probability also tends to increase as the RCS variance increases. In addition, the \textit{ISAC+S} algorithm supports lower RCS variances, about 1\,dB lower than that of \textit{ISAC}, for a specific $P_{\rm d}$.
%
%
\subsection{Impact of Target-Free Channels and Target Detector}
In this subsection, we aim to provide a comprehensive study on the impact of the target-free channels on the detection probability. First, we consider two scenarios: \textit{realistic} and \textit{idealistic}. The former corresponds to the described system model with target-free channels that cannot be canceled at the receive APs, while the latter is for the case that target-free channels can be completely suppressed. Moreover, we compare the performance of the proposed MAPRT detector with a simpler detector proposed in our previous work \cite{behdad2022power}. 
The MAPRT detector presented in 
\cite{behdad2022power} represents the scenario when the target-free channels are disregarded and is referred to as \textit{simple processing (SP)} detector here. However, the detector proposed in this paper is developed based on the presence of the target-free channels and is referred to as \textit{advanced processing (AP)} detector. Therefore, according to the combination of the \textit{realistic/idealistic} scenarios and the MAPRT detectors, we will have three cases as: a) \textit{AP-real} representing a realistic scenario with \textit{AP} detector; b) \textit{SP-real} representing a realistic scenario with \textit{SP} detector; and c) \textit{SP-ideal} representing an idealistic scenario with \textit{SP} detector. It is worth mentioning that in an idealistic scenario, the utilization of \textit{AP} detectors becomes unnecessary, as there should be no interference from target-free channels.

The effect of the unknown target-free channels on the detection probability is investigated in Fig. \ref{Pd_scalingParam}, which illustrates the detection probability versus scaling parameter, $\mathfrak{s}$, for \textit{SP} and \textit{AP} detectors. Recall that the scaling parameter $\mathfrak{s}$ affects the amount of interference from the unknown part of the target-free channels. When $\mathfrak{s}\!=\!1$, no information about the target-free channels is available, which makes them impossible to be suppressed. As expected, reducing $\mathfrak{s}$ leads to an improved detection probability, as the interference power for sensing signals decreases with a lower value of $\mathfrak{s}$. More importantly, the efficiency of \textit{AP} detector that exploits the statistics of the target-free channels in \textit{realistic} scenarios has been shown since there is a significant gap between the performance of both proposed algorithms, \textit{ISAC} and \textit{ISAC+S}, with \textit{AP} and \textit{SP}. The results also show that the proposed scheme is quite robust to the power increase corresponding to the unknown part of the target-free channel. 

\begin{figure}[t!]
	
	\begin{minipage}[t]{1\linewidth}
		\begin{center}
			\includegraphics[trim={2mm 0mm 9mm 4mm},clip,width=0.9\textwidth]{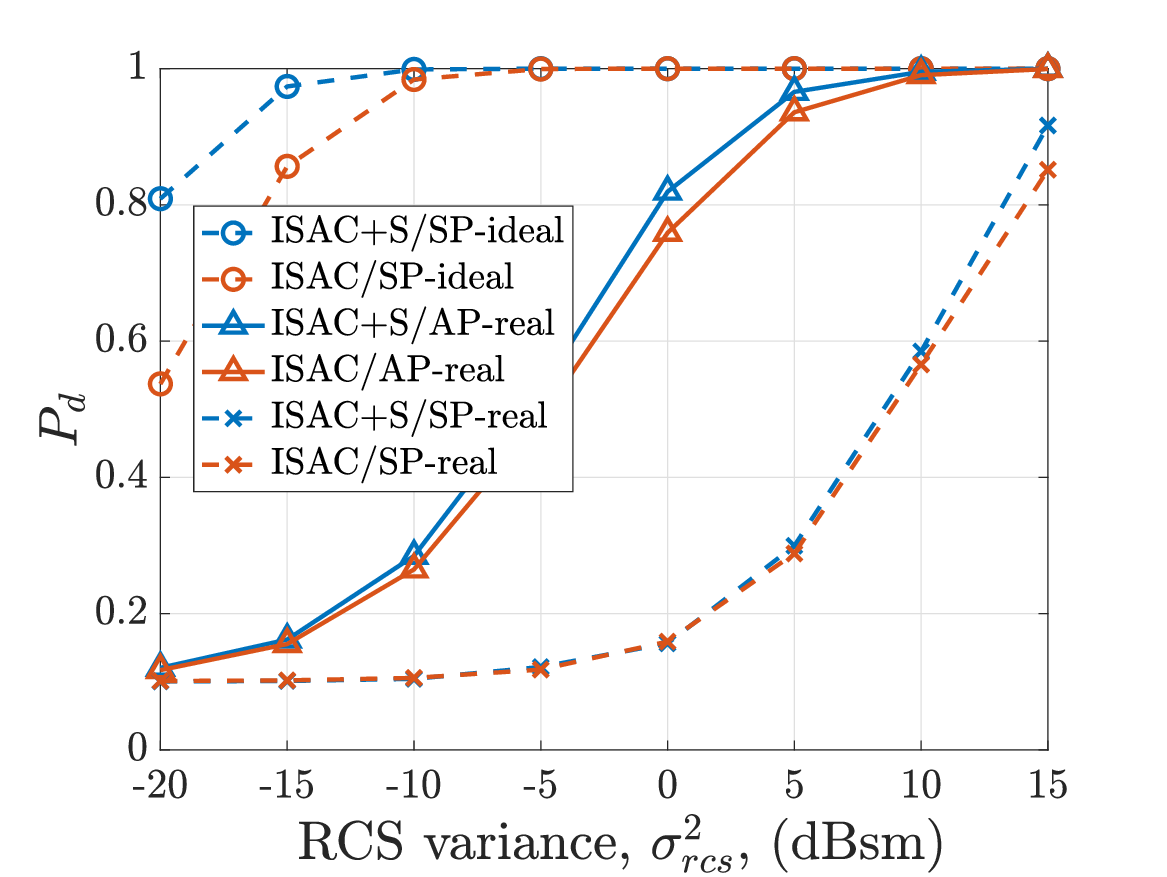}
			\caption{$P_{\rm d}$ vs. RCS variance for
 $P_{\rm fa}=0.1$ and $\mathfrak{s}=0.3$.} \label{benchmark}
		\end{center}
	\end{minipage} 
\end{figure}

%

%
Finally, the efficiency of the proposed algorithms in both \textit{realistic} and \textit{idealistic} scenarios, with \textit{SP} and \textit{AP} detectors is investigated in terms of RCS variance in Fig. \ref{benchmark}. The results demonstrate that \textit{SP-real} provides a lower bound, while the performance of \textit{SP-ideal} serves as an upper bound. In the realistic scenario, there is a substantial gap between the use of \textit{SP} and \textit{AP} detectors. However, for high RCS variances, \textit{SP} detector can also reach high performance. Moreover, a better performance of \textit{SP} detector is expected  for lower scaling parameters (see Fig.~\ref{Pd_scalingParam}). This implies that when the target has good reflectivity and the interference from unknown target-free channels is low, we can ignore the target-free channels and use the \textit{SP} detector, which is more computationally friendly than the \textit{AP} detector.
 \vspace{-1.5mm}
\section{Conclusion} \label{section9}
An ISAC framework in a cell-free massive MIMO is proposed for downlink communication with imperfect CSI, as well as, multi-static target detection in the presence of clutter. A MAPRT detector is developed to detect the target with a known location but unknown RCS using centralized processing. A power allocation algorithm is proposed to maximize the sensing SINR while satisfying both communication and per-AP power constraints to improve the ISAC performance. Furthermore, the sensing SE is defined as a novel sensing performance metric in ISAC that highlights the importance of resource utilization in ISAC. 
Numerical results show that, compared to a fully communication-centric algorithm, and in the presence of clutter, a significant improvement in detection probability can be achieved by the proposed algorithm with and without additional sensing beams. However, the use of additional sensing beams further supports lower RCS variances. Moreover, a significant gap has been observed between the performance of the proposed algorithms, which exploit the statistics of the target-free channels, and the one in which these channels are ignored in the processing. However, this gap becomes smaller when the clutter power reduces, and the RCS variance is relatively high facilitating computationally simpler target detection. Furthermore, we have observed that ISAC approaches place themselves beyond the sensing-communication SE region of an orthogonal resource-sharing approach, which emphasizes the effectiveness of the integrated operation of sensing and communication. 
\appendix
\section*{Appendix A: Proof of Lemma ~\ref{lemma:H0}}\label{appendix:A}
\subsection{Optimization under $\mathcal{H}_0$}
Given that the hypothesis $\mathcal{H}_0$ is true, we have $\textbf{y}_{\tau}= \textbf{h}_{\tau}+\textbf{n}_{\tau}$.  In order to solve the optimization problem under $\mathcal{H}_0$, we first derive the conditional  PDF of $\textbf{y}_{\tau}$ given the target-free channel vector $\boldsymbol{\mathfrak{h}}$ and the hypothesis $\mathcal{H}_0$ and PDF of $\boldsymbol{\mathfrak{h}}$ in \eqref{maxUnderH0} as
\begin{align}  &p\left(\textbf{y}_{\tau}|\boldsymbol{\mathfrak{h}},\mathcal{H}_0\right)  = p\left(\textbf{n}_{\tau}= \textbf{y}_{\tau}-\textbf{h}_{\tau}|\boldsymbol{\mathfrak{h}}\right)  \nonumber\\
&= C_1 \exp{\left(-\sum_{m=1}^{\tau_s}\frac{\left\Vert\textbf{y}[m]- \textbf{X}[m]\boldsymbol{\mathfrak{h}} \right\Vert ^2}{\sigma_n^2}\right)}, \\
&p\left(\boldsymbol{\mathfrak{h}}\right| \mathcal{H}_0) =  p\left(\boldsymbol{\mathfrak{h}}\right) =  C_2 \exp{\left(-\boldsymbol{\mathfrak{h}}^H\textbf{R}^{-1}\boldsymbol{\mathfrak{h}}\right)}, \label{p_h}  
\end{align}
respectively, where $C_1 = \frac{1}{\left(\pi \sigma_n^2\right)^{ MN_{\rm rx}\tau_s}}$ and $C_2 =\frac{1}{\pi^{M^2N_{\rm tx}N_{\rm rx}
    }\det\left(\textbf{R}\right)}$. Note that $\boldsymbol{\mathfrak{h}}$ is independent of the hypothesis $\mathcal{H}_0$.
Therefore, the optimization under $\mathcal{H}_0$ can be rewritten as 
\begin{align}
\Lambda_{\rm den}&= \max_{\boldsymbol{\mathfrak{h}}} \, C_1 C_2 \exp{\left(-\sum_{m=1}^{\tau_s}\frac{\Vert\textbf{y}[m]- \textbf{X}[m] \bm{\mathfrak{h}} \Vert ^2}{\sigma_n^2}-\bm{\mathfrak{h}}^H\textbf{R}^{-1}\bm{\mathfrak{h}}\right)}\nonumber\\
\label{min_Under_H0}
    &=C_1 C_2 \exp{\left(-\frac{1}{\sigma_n^2}\min_{\boldsymbol{\mathfrak{h}}} \,f_{\mathcal{H}_0}(\boldsymbol{\mathfrak{h}}) \right)},
\end{align}
where
\begin{align} \label{eq:fH0}
   f_{\mathcal{H}_0}(\boldsymbol{\mathfrak{h}})=
    &\bm{\mathfrak{h}}^H \left(\sum_{m=1}^{
    \tau_s}\textbf{X}^H[m] \textbf{X}[m] + \sigma_n^2 \textbf{R}^{-1}\right)\bm{\mathfrak{h}} \nonumber\\
    &-2 \Re \left(\bm{\mathfrak{h}} ^H\sum_{m=1}^{\tau_s} \textbf{X}^H[m] \textbf{y}[m]\right)+\sum_{m=1}^{\tau_s} \Vert \textbf{y}[m]\Vert^2\nonumber\\
     =& \bm{\mathfrak{h}}^H \textbf{D} \bm{\mathfrak{h}}-\bm{\mathfrak{h}}^H \textbf{b}-\textbf{b}^H\bm{\mathfrak{h}}+F,
     \end{align}
where $\textbf{D}$, $\textbf{b}$, and $F$ are given in \eqref{eq:C-D}, \eqref{eq:a-b}, and \eqref{eq:e-F}, respectively.
Then, the estimate of the vectorized target-free channel $\bm{\mathfrak{h}}$ that minimizes the quadratic convex objective function $ f_{\mathcal{H}_0}(\boldsymbol{\mathfrak{h}})$ in \eqref{min_Under_H0} is obtained as $\hat{\bm{\mathfrak{h}}} =\textbf{D}^{-1}\textbf{b}$. 
Finally, substituting the optimally estimated $\hat{\bm{\mathfrak{h}}}$ into $ f_{\mathcal{H}_0}(\boldsymbol{\mathfrak{h}})$ in \eqref{min_Under_H0}, we obtain the denominator of the likelihood ratio in \eqref{eq:likelihood-MAPRT} as in \eqref{Lambda_den}. 
\subsection{Optimization under $\mathcal{H}_1$}
Given that the hypothesis $\mathcal{H}_1$ is true, we have $\textbf{y}_{\tau}=
  \textbf{g}_{\tau}+\textbf{h}_{\tau}+\textbf{n}_{\tau}$. Then, the conditional  PDF of $\textbf{y}_{\tau}$ given $\boldsymbol{\alpha}$, $\boldsymbol{\mathfrak{h}}$ and the hypothesis $\mathcal{H}_1$ is obtained as
\begin{align}
    &p\!\left(\textbf{y}_{\tau}|\boldsymbol{\alpha},\boldsymbol{\mathfrak{h}},\mathcal{H}_1\right) \! =\! p\!\left(\textbf{n}_{\tau}\!= \!\textbf{y}_{\tau}\!-\!\textbf{g}_{\tau}\!-\!\textbf{h}_{\tau}|\boldsymbol{\alpha},\boldsymbol{\mathfrak{h}}\right)\!\nonumber\\
    &= \!C_1\!\exp{\left(-\!\sum_{m=1}^{\tau_s}\frac{\Vert\textbf{y}[m]-\textbf{G}[m]\boldsymbol{\alpha}-\textbf{X}[m]\boldsymbol{\mathfrak{h}}\!\Vert ^2}{\sigma_n^2}\right)},
\end{align}
where $\boldsymbol{\mathfrak{h}}$ and $\boldsymbol{\alpha}$ are independent of the hypothesis $\mathcal{H}_1$. Therefore, 
$ p\left(\boldsymbol{\mathfrak{h}}\right| \mathcal{H}_1) =  p\left(\boldsymbol{\mathfrak{h}}\right)$ as given in \eqref{p_h} and
the conditional PDF of $\boldsymbol{\alpha}$ given the hypothesis $\mathcal{H}_1$ is equal to its PDF given as  
\begin{align}
    p\left(\boldsymbol{\alpha}\right| \mathcal{H}_1) & = p\left(\boldsymbol{\alpha}\right) = C_3   \exp{\left(-\boldsymbol{\alpha}^H\textbf{R}_{\rm rcs}^{-1}\boldsymbol{\alpha}\right)},
\end{align}
where $C_3 =\frac{1}{\pi^{N_{\rm tx}N_{\rm rx}
    }\det\left(\textbf{R}_{\rm rcs}\right)}$. 
Now, we can rewrite \eqref{maxUnderH1} as 
\begin{align}\label{optH1}
 \Lambda_{\rm num}= C_1 C_2 C_3 \exp\left(-\frac{1}{\sigma_n^2} \min_{\boldsymbol{\alpha},\boldsymbol{\mathfrak{h}}} f_{\mathcal{H}_1}(\boldsymbol{\alpha},\boldsymbol{\mathfrak{h}})\right),
\end{align}
where
\begin{align}\
    f_{\mathcal{H}_1}(\boldsymbol{\alpha},\boldsymbol{\mathfrak{h}})=& \sum_{m=1}^{\tau_s}\Vert\textbf{y}[m]- \textbf{G}[m]\boldsymbol{\alpha}-\textbf{X}[m]\boldsymbol{\mathfrak{h}} \Vert ^2+\boldsymbol{\alpha}^H \sigma_n^2\textbf{R}_{\rm rcs}^{-1}\boldsymbol{\alpha}\nonumber\\
    +& \boldsymbol{\mathfrak{h}}^H\sigma_n^2\textbf{R}^{-1}\boldsymbol{\mathfrak{h}}\nonumber \\
    =&\sum_{m=1}^{\tau_s} \Vert \textbf{y}[m]\Vert^2 -2 \Re \left( \boldsymbol{\alpha}^H \sum_{m=1}^{\tau_s} \textbf{G}^H[m] \textbf{y}[m] \right) \nonumber\\
    & -2\Re\left(\boldsymbol{\mathfrak{h}}^H \sum_{m=1}^{\tau_s} \textbf{X}^H[m] \textbf{y}[m]\right) \nonumber \\
    &  +\boldsymbol{\alpha}^H\left(\sum_{m=1}^{\tau_s} \textbf{G}^H[m]\textbf{G}[m]+ \sigma_n^2 \textbf{R}_{\rm rcs}^{-1}\right)\boldsymbol{\alpha} \nonumber\\
    & +\boldsymbol{\mathfrak{h}}^H\left(\sum_{m=1}^{\tau_s} \textbf{X}^H[m]\textbf{X}[m]+ \sigma_n^2 \textbf{R}^{-1}\right)\boldsymbol{\mathfrak{h}} \nonumber\\
    & + 2\Re\left(\boldsymbol{\alpha}^H\sum_{m=1}^{\tau_s} \textbf{G}^H[m]\textbf{X}[m]\boldsymbol{\mathfrak{h}}\right). \label{fH1}
\end{align}
The objective function in  \eqref{fH1} can be rewritten as
\begin{align}\label{fH1Matrix}
    f_{\mathcal{H}_1}(\boldsymbol{\alpha},\boldsymbol{\mathfrak{h}})=&  \begin{bmatrix}
\boldsymbol{\alpha}^H &
\boldsymbol{\mathfrak{h}}^H
\end{bmatrix}\begin{bmatrix}
 \textbf{C}& \textbf{E} \\ 
 \textbf{E}^H&\textbf{D} 
\end{bmatrix}\begin{bmatrix}
\boldsymbol{\alpha} \\
\boldsymbol{\mathfrak{h}}
\end{bmatrix}\nonumber\\
&-2 \Re \left(\begin{bmatrix}
\boldsymbol{\alpha}^H &
\boldsymbol{\mathfrak{h}}^H
\end{bmatrix}\begin{bmatrix}
\textbf{a}\\ 
\textbf{b}
\end{bmatrix}\right)+F, 
\end{align}
where the vectors $\textbf{a}$ and $\textbf{b}$, the matrices $\textbf{C}$, $\textbf{D}$, and $\textbf{E}$, and the scalar $F$ are given in \eqref{eq:a-b}, \eqref{eq:C-D}, and \eqref{eq:e-F}, respectively.
Equating the first derivative of $f_{\mathcal{H}_1}(\boldsymbol{\alpha},\boldsymbol{\mathfrak{h}})$ with respect to $\boldsymbol{\alpha}^H$ and $\boldsymbol{\mathfrak{h}}^H$ 
to zero, the estimated  $\boldsymbol{\alpha}$ and $\boldsymbol{\mathfrak{h}}$ are obtained as
\begin{align}\label{estimation}
    \begin{bmatrix}
\hat{\boldsymbol{\alpha}}\\ 
\hat{\boldsymbol{\mathfrak{h}}}
\end{bmatrix}=\begin{bmatrix}
 \textbf{C}& \textbf{E} \\ 
 \textbf{E}^H&\textbf{D} 
\end{bmatrix}^{-1}\begin{bmatrix}
\textbf{a}\\ 
\textbf{b}
\end{bmatrix}.
\end{align}
Substituting \eqref{estimation} in \eqref{fH1Matrix}, we have
\begin{align}\label{min_fH1}
   \min_{\boldsymbol{\alpha},\boldsymbol{\mathfrak{h}}} f_{\mathcal{H}_1}(\boldsymbol{\alpha},\boldsymbol{\mathfrak{h}})=&  - \begin{bmatrix}
\textbf{a}^H &
\textbf{b}^H
\end{bmatrix} \begin{bmatrix}
 \textbf{C}& \textbf{E} \\ 
 \textbf{E}^H&\textbf{D} 
\end{bmatrix}^{-1}\begin{bmatrix}
\textbf{a}\\ 
\textbf{b}
\end{bmatrix}+F.
\end{align}
Then, the numerator in \eqref{eq:likelihood-MAPRT} is derived by inserting \eqref{min_fH1} into \eqref{optH1} and given in \eqref{Lambda_num}.


\section*{Appendix B: Proof of Sensing SINR}\label{app:proof_SNR}

The sensing SINR for the received vector signal in \eqref{hypothesis} is given as
\begin{align}\label{gamma_s}
    \mathsf{SINR}_{\rm s}^{\rm isac} &= \frac{\mathbb{E}\left\{\Vert \textbf{g}_{\tau}\Vert^2\right\}}{\mathbb{E}\left\{\Vert \textbf{n}_{\tau}\Vert^2\right\}+\mathbb{E}\left\{\Vert \textbf{h}_{\tau}\Vert^2\right\}} \nonumber\\
    &= \frac{\sum_{m=1}^{\tau_s}\mathbb{E}\left\{\Vert \textbf{G}[m]\boldsymbol{\alpha}\Vert^2\right\}}{\tau_s MN_{\rm rx}\sigma_n^2+\sum_{m=1}^{\tau_s}\mathbb{E}\left\{\Vert \textbf{X}[m]\boldsymbol{\mathfrak{h}}\Vert^2\right\}}.
\end{align}
Using \eqref{y_rPrim}, the expectation $\mathbb{E}\left\{\Vert \textbf{G}[m]\boldsymbol{\alpha}\Vert^2\right\}$ in the numerator  can be computed as

\begin{align}
 &  \mathbb{E}\left\{\Vert \textbf{G}[m]\boldsymbol{\alpha}\Vert^2\right\}  
   = \sum_{r=1}^{N_{\rm rx}}\mathbb{E}\left\{\left\Vert\sum_{k=1}^{N_{\rm tx}}\alpha_{r,k}\textbf{g}_{r,k}[m]\right\Vert^2\right\} 
   \nonumber\\
   &=\sum_{r=1}^{N_{\rm rx}}\sum_{k=1}^{N_{\rm tx}}\sum_{j=1}^{N_{\rm tx}}\mathbb{E}\left\{ \textbf{g}_{r,k}^H[m]\alpha_{r,k}^*\alpha_{r,j}\textbf{g}_{r,j}[m]\right\}\nonumber\\
 & \stackrel{(a)}{=}\sum_{r=1}^{N_{\rm rx}}\sum_{k=1}^{N_{\rm tx}}\sum_{j=1}^{N_{\rm tx}}\textbf{g}_{r,k}^H[m]\textrm{cov}\left( \alpha_{r,j},\alpha_{r,k}\right)\textbf{g}_{r,j}[m]
  \nonumber\\
  &\stackrel{(b)}{=}\boldsymbol{\rho}^T \textbf{D}_{\rm s}^H[m] \Bigg(\sum_{r=1}^{N_{\rm rx}}\sum_{k=1}^{N_{\rm tx}}\sum_{j=1}^{N_{\rm tx}}\sqrt{\beta_{r,k}\beta_{r,j}}\textrm{cov}\left( \alpha_{r,j},\alpha_{r,k}\right)\textbf{W}_k^H\nonumber\\
 &\hspace{3mm}\times\textbf{a}^{*}(\varphi_{k},\vartheta_{k})\textbf{a}^H(\phi_{r},\theta_{r})\textbf{a}(\phi_{r},\theta_{r})\textbf{a}^{T}(\varphi_{j},\vartheta_{j})  \textbf{W}_j\Bigg)\textbf{D}_{\rm s}[m]\boldsymbol{\rho},
\end{align}
where in $(a)$, we treat $\textbf{g}_{r,k}[m]$ as deterministic since the edge cloud knows $\textbf{x}_k[m]$. In $(b)$, we recalled the definition of $\textbf{g}_{r,k}[m]$ from \eqref{y_rPrim} and $\textbf{x}_k[m]$ from \eqref{x_k}. Defining 
\begin{align}
    \textbf{A}=& \sum_{m=1}^{\tau_s}\textbf{D}_{\rm s}^H[m] \Bigg(\sum_{r=1}^{N_{\rm rx}}\sum_{k=1}^{N_{\rm tx}}\sum_{j=1}^{N_{\rm tx}}\sqrt{\beta_{r,k}\beta_{r,j}}\textbf{W}_k^H\textbf{a}^{*}(\varphi_{k},\vartheta_{k})\textbf{a}^H(\phi_{r},\theta_{r})\nonumber\\
 &\times\textrm{cov}\left( \alpha_{r,j},\alpha_{r,k}\right)\textbf{a}(\phi_{r},\theta_{r})\textbf{a}^{T}(\varphi_{j},\vartheta_{j})  \textbf{W}_j\Bigg)\textbf{D}_{\rm s}[m],
 \end{align}
the term $\sum_{m=1}^{\tau_s}\mathbb{E}\left\{\Vert \textbf{G}[m]\boldsymbol{\alpha}\Vert^2\right\}$ in the numerator of the sensing SINR in \eqref{gamma_s} can be written as $\sum_{m=1}^{\tau_s}\mathbb{E}\left\{\Vert \textbf{G}[m]\boldsymbol{\alpha}\Vert^2\right\} = \boldsymbol{\rho}^T\textbf{A}\boldsymbol{\rho}.$


Using \eqref{y_rPrim}, $\mathbb{E}\left\{\Vert \textbf{X}[m]\boldsymbol{\mathfrak{h}}\Vert^2\right\}$ in the denominator of the sensing SINR in \eqref{gamma_s} can be expressed as
\begin{align} \label{ex:Xh-frak}
   \mathbb{E}\left\{\Vert \textbf{X}[m]\boldsymbol{\mathfrak{h}}\Vert^2\right\}&= \sum_{r=1}^{N_{\rm rx}} \mathbb{E}\left\{ \left\Vert \sum_{k=1}^{N_{\rm tx}} \textbf{H}_{r,k}\textbf{x}_k[m]
   \right \Vert^2 \right\} \nonumber\\
   &\stackrel{(a)}{=}\sum_{r=1}^{N_{\rm rx}}\sum_{k=1}^{N_{\rm tx}} \sum_{j=1}^{N_{\rm tx}} \textbf{x}^H_k[m] \mathbb{E}\left\{\textbf{H}_{r,k}^H \textbf{H}_{r,j}\right\} \textbf{x}_j[m] \nonumber \\
    & \stackrel{(b)}{=} \sum_{r=1}^{N_{\rm rx}}\sum_{k=1}^{N_{\rm tx}} \textbf{x}^H_k[m] \mathbb{E}\left\{\textbf{H}_{r,k}^H \textbf{H}_{r,k}\right\} \textbf{x}_k[m] \nonumber\\&\stackrel{(c)}{=}\boldsymbol{ \rho}^T\textbf{D}^H_{\rm s}[m]\! \left(\!\sum_{r=1}^{N_{\rm rx}}\sum_{k=1}^{N_{\rm tx}}\!\textbf{W}_k^H \boldsymbol{\mathfrak{R}}_{r,k}\textbf{W}_k\!\right)\! \textbf{D}_{\rm s}[m]\boldsymbol{ \rho},
\end{align}
where we treat $\textbf{x}_k[m]$ as deterministic in $(a)$ since it is known in the edge cloud. In $(b)$, we have used that zero-mean channels $\textbf{H}_{r,k}$ are independent for different $k$ and in $(c)$ we recalled the definition of $\textbf{x}_k[m]$ from \eqref{x_k} and defined $\boldsymbol{\mathfrak{R}}_{r,k} \triangleq\mathbb{E}\left\{\textbf{H}_{r,k}^H \textbf{H}_{r,k}\right\} $. Using the correlated Rayleigh modeling of the channel $\textbf{H}_{r,k}$ from \eqref{eq:Hrk}, the defined expectation term  $\boldsymbol{\mathfrak{R}}_{r,k}$ can be written as \eqref{eq:R-frak}, where we use the fact that the square root of the positive semi-definite matrix  $\textbf{R}_{{\rm rx},(r,k)}$ is Hermitian symmetric.

\begin{figure*}[!t]
   \begin{align} \label{eq:R-frak}
   \boldsymbol{\mathfrak{R}}_{r,k}
  &= \mathbb{E}\left\{ \left(\textbf{R}^{\frac{1}{2}}_{{\rm tx},(r,k)}\right)^* \textbf{W}^H_{{\rm ch},{(r,k)}}\left(\textbf{R}^{\frac{1}{2}}_{{\rm rx},(r,k)}\right)^H
  \textbf{R}^{\frac{1}{2}}_{{\rm rx},(r,k)} \textbf{W}_{{\rm ch},{(r,k)}}\left(\textbf{R}^{\frac{1}{2}}_{{\rm tx},(r,k)}\right)^T \right\} \nonumber\\
  & = \left(\textbf{R}^{\frac{1}{2}}_{{\rm tx},(r,k)}\right)^*\mathbb{E}\left\{  \textbf{W}^H_{{\rm ch},{(r,k)}}
  \textbf{R}_{{\rm rx},(r,k)} \textbf{W}_{{\rm ch},{(r,k)}} \right\}\left(\textbf{R}^{\frac{1}{2}}_{{\rm tx},(r,k)}\right)^T,
\end{align}
\begin{align} \label{eq:exp-trace}
 &\mathbb{E}\left\{  \left[\textbf{W}_{\mathrm{ch},(r,k)}\right]_{:m}^H\textbf{R}_{{\rm rx},(r,k)}\left[\textbf{W}_{\mathrm{ch},(r,k)}\right]_{:n}  \right\} 
=\!\mathrm{tr}\left(\textbf{R}_{{\rm rx},(r,k)}\mathbb{E}\!\left\{\left[\textbf{W}_{\mathrm{ch},(r,k)}\right]_{:n}\!\left[\textbf{W}_{\mathrm{ch},(r,k)}\right]_{:m}^H\!\right\}\right) 
\end{align}
\begin{align}\label{eq:R_r,k}
    \boldsymbol{\mathfrak{R}}_{r,k} = \mathrm{tr}\left(\textbf{R}_{\mathrm{rx},(r,k)}\right) \left(\textbf{R}^{\frac{1}{2}}_{{\rm tx},(r,k)}\right)^*\left(\textbf{R}^{\frac{1}{2}}_{{\rm tx},(r,k)}\right)^T =\mathrm{tr}\left(\textbf{R}_{\mathrm{rx},(r,k)}\right)\textbf{R}^{T}_{{\rm tx},(r,k)}.
\end{align}
\begin{align}\label{eq:expectation}
     \mathbb{E}\left\{\Vert \textbf{X}[m]\boldsymbol{\mathfrak{h}}\Vert^2\right\}=\boldsymbol{ \rho}^T\textbf{D}^H_{\rm s}[m] \left(\sum_{r=1}^{N_{\rm rx}}\sum_{k=1}^{N_{\rm tx}}\mathrm{tr}\left(\textbf{R}_{\mathrm{rx},(r,k)}\right)\textbf{W}_k^H \textbf{R}^{T}_{{\rm tx},(r,k)} \textbf{W}_k\right) \textbf{D}_{\rm s}[m]\boldsymbol{ \rho}.  
\end{align}
\hrulefill
\end{figure*}
The $(m,n)$th entry of the matrix $\mathbb{E}\left\{  \textbf{W}^H_{{\rm ch},{(r,k)}}
  \textbf{R}_{{\rm rx},(r,k)} \textbf{W}_{{\rm ch},{(r,k)}} \right\}$ is computed as \eqref{eq:exp-trace}, where we have used the trace operation and its cyclic shift property. In the above expression, the expectation $\mathbb{E}\left\{\left[\textbf{W}_{\mathrm{ch},(r,k)}\right]_{:n} \left[\textbf{W}_{\mathrm{ch},(r,k)}\right]_{:m}^H \right\}$ is zero if $m\neq n$. Otherwise, if $m=n$, it is equal to the identity matrix. So, the matrix $\mathbb{E}\left\{  \textbf{W}^H_{{\rm ch},{(r,k)}}
  \textbf{R}_{{\rm rx},(r,k)} \textbf{W}_{{\rm ch},{(r,k)}} \right\}$ in \eqref{eq:R-frak} is $\mathrm{tr}\left(\textbf{R}_{\mathrm{rx},(r,k)}\right)\textbf{I}_M$. 
  Inserting this into \eqref{eq:R-frak}, we obtain \eqref{eq:R_r,k}
.
Then, the expression in \eqref{ex:Xh-frak}  is obtained as \eqref{eq:expectation}.
Now, the second term in the denominator in \eqref{gamma_s} can be expressed as $\boldsymbol{\rho}^T\textbf{B}\boldsymbol{\rho}$, where
\begin{align}
\textbf{B}=\sum_{m=1}^{\tau_s}\textbf{D}_{\rm s}^H[m] \left(\sum_{r=1}^{N_{\rm rx}}\sum_{k=1}^{N_{\rm tx}}\mathrm{tr}\left(\textbf{R}_{\mathrm{rx},(r,k)}\right)\textbf{W}_k^H \textbf{R}^{T}_{{\rm tx},(r,k)}\textbf{W}_k \right)\textbf{D}_{\rm s}[m].
\end{align}
Then, we obtain the sensing SINR as 
\begin{equation}
     \mathsf{SINR}_{\rm s}^{\rm isac}=\frac{\boldsymbol{\rho}^T \textbf{A} \boldsymbol{\rho}}{\tau_s MN_{\rm rx}\sigma_n^2+ \boldsymbol{\rho}^T\textbf{B}\boldsymbol{\rho} }. 
\end{equation}


\section*{Appendix C: Proof of Convexity of $f\left(\boldsymbol{\rho},t\right)$}
\label{appendix:C}
  We want to show that $f\left(\boldsymbol{\rho},t\right)$ is a convex function. First, we derive the gradient vector of function $f\left(\boldsymbol{\rho},t\right)$ as follows
 \begin{align}
     \nabla f\left ( \boldsymbol{\rho},t \right )= \begin{bmatrix}
     \frac{2\textbf{A}_r\boldsymbol{\rho}}{t} \\ -\frac{\boldsymbol{\rho}^T\textbf{A}_r\boldsymbol{\rho}}{t^2}
     \end{bmatrix}.
 \end{align}
The Hessian matrix would be 
\begin{align}
     \textbf{H}\left ( \boldsymbol{\rho},t \right )= \begin{bmatrix}
     \frac{2\textbf{A}_r}{t} &-\frac{2\textbf{A}_r\boldsymbol{\rho}}{t^2}\\ -\frac{2\boldsymbol{\rho}^T\textbf{A}_r}{t^2} & \frac{2\boldsymbol{\rho}^T\textbf{A}_r\boldsymbol{\rho}}{t^3}
     \end{bmatrix}.
 \end{align}
 By showing that the Hessian matrix is positive semi-definite, we prove that the objective function $f\left(\boldsymbol{\rho},t \right)$ is convex. The former is equivalent to showing the non-negativity of $\begin{bmatrix}
 \bar{\boldsymbol{\rho}}^T & \bar{t}
 \end{bmatrix}  \textbf{H}\left ( \boldsymbol{\rho},t \right ) \begin{bmatrix}
 \bar{\boldsymbol{\rho}} \\ \bar{t}
 \end{bmatrix}$ for any selection of $\bar{\boldsymbol{\rho}}$ and $\bar{t}$. This product can be computed in a quadratic form as 
\begin{align}
  & \begin{bmatrix}
 \bar{\boldsymbol{\rho}}^T & \bar{t}
 \end{bmatrix} \! \begin{bmatrix}
     \frac{2\textbf{A}_r}{t} \!&\!-\frac{2\textbf{A}_r\boldsymbol{\rho}}{t^2}\\ -\frac{2\boldsymbol{\rho}^T\textbf{A}_r}{t^2} \!&\! \frac{2\boldsymbol{\rho}^T\textbf{A}_r\boldsymbol{\rho}}{t^3}
     \end{bmatrix}\! \begin{bmatrix}
 \bar{\boldsymbol{\rho}} \\ \bar{t}
 \end{bmatrix} 
 \!=\!2\frac{\left(\bar{\boldsymbol{\rho}}-\frac{\bar{t}}{t}\boldsymbol{\rho}\right)^T\textbf{A}_r\left(\bar{\boldsymbol{\rho}}-\frac{\bar{t}}{t}\boldsymbol{\rho}\right)}{t},
 \end{align}
 which is $\geq 0$. Since $\textbf{A}_r$ is positive semi-definite, the above result is always non-negative.
As a result, we can conclude that the Hessian matrix is positive semi-definite and the objective function is convex.

\bibliographystyle{IEEEtran}
\bibliography{IEEEabrv,refs}
\end{document}